%% file: main.tex
\documentclass[10]{vldb}	
\input{UsefulLatexStuff} 
\hypersetup{pdftex,
            pdfauthor={Wolfgang Gatterbauer and Dan Suciu},
            pdftitle={Approximate Lifted Inference in Probabilistic Databases},
            pdfkeywords={Probabilistic databases, Weighted model counting, Boolean expressions, Probabilistic inference, Relaxation}}

\global\boilerplate={}
\global\conf{}
\global\confinfo{}
\global\copyrightetc{}
\global\copyrightetc{Pre-print for a paper appearing in VLDB 2015}

\begin{document}
\title{Approximate Lifted Inference with Probabilistic Databases}

\author{
	Wolfgang Gatterbauer\\
	\affaddr{\normalsize{Carnegie Mellon University}}\\
	\textsf{gatt@cmu.edu}
 	\and 
	Dan Suciu \\
	\affaddr{\normalsize{University of Washington}}\\
	\textsf{suciu@cs.washington.edu}
}
\maketitle

\sloppy
\begin{abstract}
\input{0abstract}

\end{abstract}

\input{1introduction}

\input{2background}

\input{3-1queryDissociation}

\input{3-2queryDissociationPlans}

\input{3-3schemaKnowledge}

\input{4optimizingDissociation}

\input{5experiments}

\input{6relatedWork}

\input{7conclusion}

\smallsection{Acknowledgements}
This work was supported in part by NSF grants IIS-0513877, IIS-0713576, IIS-0915054, and IIS-1115188.
We thank the reviewers for 
their careful reading 
of this manuscript 
and 
their detailed feedback. WG
would also like to thank Manfred Hauswirth for a small comment
in 2007 that was crucial for 
the development of the ideas in this paper.

\balance
\bibliography{\bibpath}

\end{document}

%% file: UsefulLatexStuff.tex
\usepackage{balance}
\usepackage{mathptmx}      	
\usepackage{mathtools}	
\usepackage{graphicx}     	
\usepackage{multirow}       
\usepackage{amsmath}				
\usepackage{amssymb}                
\usepackage{bm}			
\usepackage{upgreek}       

\usepackage{caption}                
\usepackage{xspace}            
\usepackage[caption=false]{subfig}   

\usepackage{enumitem}               

\usepackage{color}	
\usepackage{colortbl}             

\usepackage{tabularx}		

\usepackage{rotating}

\usepackage{aliascnt}

\usepackage{ifpdf}
    \ifpdf                                                                  
        \usepackage[breaklinks=true]{hyperref}				
        \hypersetup{
            colorlinks=false,               
            citebordercolor={0 1 0},        
            urlbordercolor={0 0 1},         
            linkbordercolor={1 0.2 0.2},    
            pdfborder={0 0 1.2},            
        }
    \else
        \usepackage[hypertex, breaklinks=true]{hyperref}       
        \hypersetup{
            colorlinks=true,               
            anchorcolor=yellow,
            linkcolor=green,
            filecolor=green,
            urlcolor=blue,
            citecolor=red,
        }
    \fi
    \hypersetup{                                                            
		pdfpagemode=UseNone,               
        pdfstartview=Fit,
        pdfpagelayout=SinglePage,       
        bookmarks=false,
        bookmarksnumbered=true,
        bookmarksopen=false,             
        pdftex=true,
        unicode,
    }

\usepackage[all]{hypcap}

\let\footnotesize\scriptsize

\newtheorem{theorem}{Theorem}          	
\newaliascnt{lemma}{theorem}				
\newtheorem{lemma}[lemma]{Lemma}              	
\aliascntresetthe{lemma}  					

\newaliascnt{conjecture}{theorem}			
    
\aliascntresetthe{conjecture}  				

\newaliascnt{remark}{theorem}				
              
\aliascntresetthe{remark}  					
\newaliascnt{corollary}{theorem}			
\newtheorem{corollary}[corollary]{Corollary}      
\aliascntresetthe{corollary}  				

\newaliascnt{definition}{theorem}			
\newtheorem{definition}[definition]{Definition}    
\aliascntresetthe{definition}  				

\newaliascnt{proposition}{theorem}			
\newtheorem{proposition}[proposition]{Proposition}  
\aliascntresetthe{proposition}  				
\newaliascnt{example}{theorem}			
\newtheorem{example}[example]{Example}  	
\aliascntresetthe{example}  				
\newaliascnt{problem}{theorem}			
  	
\aliascntresetthe{problem}  				

\newtheorem{question}{Question}

\usepackage[usenames,table,dvipsnames]{xcolor}      

\newcounter{result}

\newsavebox{\coloredbgbox}
\newenvironment{result}{
	\begin{lrbox}{\coloredbgbox}
		\begin{minipage}{\dimexpr1\linewidth-2\fboxsep-2\fboxrule\relax}
			\refstepcounter{result}
			\emph{Result \theresult.}
			\itshape
		}
  		{
		\end{minipage}
	\end{lrbox}
	\begin{center}
		\fcolorbox{black}{black!05}{\usebox{\coloredbgbox}}
	\end{center}
}

\newcommand{\specificref}[2]{\hyperref[#2]{#1~\ref*{#2}}}			
\newcommand{\specialref}[3]{\hyperref[#2]{#1~\ref*{#2}~#3}}

\newcommand{\figref}[2]{\hyperref[#1]{Fig.~\ref*{#1}#2}}	
\newcommand{\Figref}[2]{\hyperref[#1]{Figure~\ref*{#1}#2}}	

\newcommand{\figuresref}[2]{\hyperref[#1]{Fig.~\ref*{#1} to \ref*{#2}}}	
\newcommand{\Figuresref}[2]{\hyperref[#1]{Figures~\ref*{#1} to \ref*{#2}}}

\usepackage[algo2e, boxed, vlined]{algorithm2e}		
\DontPrintSemicolon     		
\SetNlSty{}{}{}					
\SetAlgoCaptionSeparator{\,}	
\SetAlgoSkip{smallskip}

\SetAlgoCaptionLayout{myCap}					

\SetAlCapNameSty{textbf}
\SetAlCapNameFnt{\small}
\SetAlgoCaptionSeparator{:}

\SetInd{1.2mm}{1.2mm}	

\SetKwFor{ForAll}{forall}{do}{endfch}	
\LinesNumbered

\newcommand{\PP}[1]{{{\mathbb{P}}}\!\left( #1 \right)}       
\newcommand{\PPP}{{{\mathbb{P}}}}       
\newcommand{\ADom}{{\mathit{ADom}}}

\renewcommand{\epsilon}{\varepsilon}

\newcommand{\smallsection}[1]{\vspace{5mm}\noindent\textbf{#1.}}	
\newcommand{\introparagraph}[1]{\textbf{#1.}}        

\definecolor{gray}{rgb}{0.5,0.5,0.5}
\definecolor{niceblue}{rgb}{.8,.85,1}

\newcommand{\sql}[1]{\textup{\textsf{\small#1}}}

\newcommand{\datarule}{{\,:\!\!-\,}}			
\newcommand{\Var}{\textup{\texttt{Var}}}		
\newcommand{\HVar}{\textup{\texttt{HVar}}} 		
\newcommand{\JVar}{\textup{\texttt{JVar}}}

\newcommand{\SepVar}{\textup{\texttt{SVar}}}

\newcommand{\EVar}{\textup{\texttt{EVar}}}  	
  	
\newcommand{\score}{\textit{score}}  			
\newcommand{\TopSet}{\textup{\texttt{MinCuts}}}  
\newcommand{\MinCuts}{\textup{\texttt{MinCuts}}}  
\newcommand{\MinPCuts}{\textup{\texttt{MinPCuts}}}

\newcommand{\at}{\textup{\texttt{at}}}  		

\newcommand{\joinp}[2]{\Join^p\!\! \big[#2\big]} 
		
\newcommand{\joind}[2]{\Join_{#1}\!\! \big[#2\big]}			
\newcommand{\projp}[1]{\pi^p_{\!-\!#1}}						
\newcommand{\projpd}[1]{\pi^p_{#1}}						 	
\newcommand{\projd}[1]{\pi_{\!-\!#1}}						
\newcommand{\projdd}[1]{\pi_{#1}}							
\newcommand{\minp}[1]{\min \! \left[ #1 \right]}

\newcommand{\avg}{\textrm{avg}}			

\renewcommand{\vec}[1]{\boldsymbol{\mathbf{#1}}}

\usepackage{scalefnt}	
\makeatletter       	
    \def\url@leostyle{
        \@ifundefined{selectfont}{\def\UrlFont{\sf}}{\def\UrlFont{\scalefont{0.9}\ttfamily}}}
\makeatother
\newcommand{\set}[1]{\{#1\}}

\newcommand{\makeop}[2]                         
  {\ifx#2.\def\next##1{}\else\escapechar=-1     
  \def\next##1{\escapechar=92\def#2{#1}}        
  \expandafter\next\expandafter{\string#2}      
  \let\next\makeop\fi\next{#1}}                 

\makeop{{\cal #1}}                              
  \W  .             

\def \var(#1){{\bf #1}}

\newcommand{\silentreminder}[1]{}

\newcommand{\eat}[1]{}

\newcommand{\bibpath}{propagation} 		
\graphicspath{{figs/}}

\renewenvironment{thebibliography}[1]{%
    \begin{oldthebibliography}{#1}%
    	\scriptsize
        \setlength{\parskip}{0.1ex}
        \setlength{\itemsep}{0.0ex}%
		\setlength{\itemindent}{0ex}		
}%
{%
\end{oldthebibliography}%
}

\renewcommand{\bibindent}{-1.8ex}					

\newcommand\xqed[1]{%
  \leavevmode\unskip\penalty9999 \hbox{}\nobreak\hfill
  \quad\hbox{#1}}
\newcommand\markend{\xqed{$\blacksquare$}} 		
\newif\ifqed

\makeatletter 

\renewcommand{\@opargbegintheorem}[3]{%
    \parskip 0pt 
    \trivlist
    \item[%
    	\hskip 0\p@										
        \hskip \labelsep
        {\sc #1\ #2\             %
   \setbox\@tempboxa\hbox{(#3)}  %
        \ifdim \wd\@tempboxa>\z@ %
            \hskip 0\p@\relax    
            \box\@tempboxa       %
        \fi.}%
    ]
    \it
}

\def\@begintheorem#1#2{%
    \parskip 0pt 
    \trivlist
    \item[%
    	\hskip 0\p@										
        \hskip \labelsep
        {{\sc #1}\hskip 5\p@\relax#2.}%
    ]
    \it
}

\makeatother

\let\orgautoref\autoref                         		
\providecommand{\Autoref}[1]{
    \def\equationautorefname{Equation}
    \def\figureautorefname{Figure}%
	\def\subfigureautorefname{Figure}%
    \def\lemmaautorefname{Lemma}%
    \def\conjectureautorefname{Conjecture}%
    \def\remarkautorefname{Remark}%
    \def\propositionautorefname{Proposition}%
    \def\corollaryautorefname{Corollary}%
    \def\definitionautorefname{Definition}%
    \def\sectionautorefname{Section}%
    \def\subsectionautorefname{Section}%
    \def\subsubsectionautorefname{Section}%
    \def\exampleautorefname{Example}%
    \def\lessonautorefname{Result}%
    \def\resultautorefname{Result}%
    \orgautoref{#1}%
}
\renewcommand{\autoref}[1]{
    \def\figureautorefname{Fig.\!}%
    \def\subfigureautorefname{Fig.\!}
    \def\lemmaautorefname{Lemma}%
    \def\conjectureautorefname{Conjecture}%
    \def\remarkautorefname{Remark}%
    \def\propositionautorefname{Prop.}%
    \def\algorithmautorefname{Algorithm}%
    \def\sectionautorefname{Section}%
    \def\subsectionautorefname{Section}%
    \def\exampleautorefname{Example}%
    \def\lessonautorefname{Result}%
    \def\resultautorefname{Result}%
    \orgautoref{#1}%
}

%% file: 0abstract.tex
This paper proposes a new approach for approximate evaluation of \#P-hard queries with probabilistic databases. In our approach, every query is evaluated entirely in the database engine by evaluating a fixed number of query plans, each providing an upper bound on the true probability, then taking their minimum. 
We provide an algorithm that takes into account important schema information to enumerate only the minimal necessary plans among all possible plans.
Importantly, this algorithm 
is a strict generalization of all known results of PTIME self-join-free conjunctive queries: A query is safe if and only if our algorithm returns one single plan.
We also apply three relational query optimization techniques to evaluate all minimal safe plans 
very fast.
We give a detailed experimental evaluation of our approach and, in the process,
provide a new way of thinking about the value of probabilistic methods over non-probabilistic methods for ranking query answers.

%% file: 1introduction.tex
\section{Introduction}\label{sec:intro}

\noindent
Probabilistic inference over large data sets is becoming a central
data management problem. Recent large knowledge bases, such as
Yago~\cite{DBLP:journals/ai/HoffartSBW13},
Nell~\cite{DBLP:conf/aaai/CarlsonBKSHM10}, DeepDive~\cite{deepdive},
or Google's Knowledge Vault~\cite{knoweldge-vault-kdd-2014}, have
millions to billions of uncertain tuples.  Data sets with missing
values are often ``completed'' using inference in graphical
models~\cite{DBLP:conf/sigmod/ChenW14,DBLP:conf/icde/StoyanovichDMT11}
or sophisticated low rank matrix factorization
techniques~\cite{bouchard-uai2014,DBLP:conf/kdd/SinghG08}, which
ultimately results in a large, probabilistic database. 
Data sets that use crowdsourcing are also 
uncertain~\cite{DBLP:conf/dasfaa/AmarilliAM14}.
And, very recently, probabilistic databases have been applied to
bootstrapping over samples of data~\cite{DBLP:conf/sigmod/ZengGMZ14}.

However, probabilistic inference is known to be \#P-hard in the size of the
database, even for some very simple
queries~\cite{DBLP:journals/vldb/DalviS07}.  Today's state of the art
inference engines use either sampling-based methods or are based on
some variant of the DPLL algorithm for Weighted Model Counting.  For example,
Tuffy~\cite{DBLP:journals/pvldb/NiuRDS11}, a popular implementation of
Markov Logic Networks (MLN) over relational databases, uses Markov Chain Monte Carlo methods (MCMC).
Gibbs sampling can be significantly improved by adapting some
classical relational optimization
techniques~\cite{DBLP:conf/sigmod/ZhangR13}.  For another example,
MayBMS~\cite{DBLP:conf/icde/AntovaKO07a} and its successor
Sprout~\cite{DBLP:conf/icde/OlteanuHK09} use query plans to guide a
DPLL-based algorithm for Weighted Model
Counting~\cite{DBLP:series/faia/GomesSS09}.  While both approaches
deploy some advanced relational optimization techniques, at their core
they are based on general purpose probabilistic inference techniques,
which either run in exponential time (DPLL-based algorithms have been
proven recently to take exponential time even for queries computable
in polynomial time~\cite{DBLP:conf/icdt/BeameLRS14}), 
or require many
iterations until convergence.

In this paper, we propose a different approach to query evaluation with
probabilistic databases.  In our approach, \emph{every query is
evaluated entirely in the database engine}.  Probability computation is
done at query time, using simple arithmetic operations and aggregates.
Thus, probabilistic inference is entirely reduced to a standard query
evaluation problem with aggregates. There are no iterations and no
exponential blowups. All benefits of relational engines (such as
cost-based optimizations, multi-core query processing, shared-nothing
parallelization) are directly available to queries over
probabilistic databases.  To achieve this, we compute approximate rather than exact probabilities, with a one-sided
guarantee: The probabilities are guaranteed to be upper bounds to the
true probabilities, which we show is \emph{sufficient to rank the top query
answers with high precision}.
Our approach consists of approximating the true query probability by
evaluating a fixed number of ``safe queries'' (the number depends on the query), each providing an upper
bound on the true probability, then taking their minimum.  

We briefly review ``safe queries,'' which are queries whose data
complexity is in PTIME.  They can be evaluated using safe query
plans~\cite{DBLP:journals/vldb/DalviS07,FinkO:PODS2014dichotomy,DBLP:series/synthesis/2011Suciu},
which are related to a technique called \emph{lifted inference} in the
AI literature~\cite{DBLP:conf/ijcai/BroeckTMDR11,jaeger-broeck-2012};
the entire computation is pushed inside the database engine and is thus efficient.
For example, the query $q_1(z)
\datarule R(z,x),S(x,y), K(x,y)$ has the safe query plan
$P_1 = \pi_z (R \Join_x (\pi_x (S \Join_{x,y} K)))$,
where every join
operator multiplies the probabilities, and every projection with
duplicate elimination treats probabilistic events as independent.
The literature describes several classes of safe
queries~\cite{DBLP:journals/jacm/DalviS12,FinkO:PODS2014dichotomy}
and shows that they can be evaluated very efficiently.
However, most queries are ``unsafe:" They are provably
\#P-hard and do not admit safe plans.

In this paper, we prove that
every conjunctive query without self-joins can be approximated by a
fixed number of safe queries, called ``\emph{safe dissociations}'' of the
original query.  
Every safe dissociation is guaranteed to return an upper
bound on the true probability and can be evaluated in PTIME data complexity.  
The number of safe dissociations depends
only on the query and not the data. 
Moreover, we show how to find ``\emph{minimal safe dissociations}'' which are sufficient to find the best
approximation to the given query.  
For example,
the unsafe query 
$q_2(z) \datarule R(z,x),S(x,y),T(y)$ has two minimal safe dissociations,
$q_2'(z) \datarule R(z,x),S(x,y), T'(x,y)$ and 
$q_2''(z) \datarule R'(z,x,y),S(x,y),T(y)$.
Both queries are safe and, by setting the probability of every tuple
$R'(z,x,y)$ equal to that of $R(z,x)$ and similarly for $T'$, they
return an upper bound for the probabilities of each answer tuple from
$q_2(z)$.  One benefit of our approach is that, if the query happens
to be safe, then it has a unique minimal safe dissociation, and our algorithm finds it.

\introparagraph{Contributions} 
(1) We show that there exists a
1-to-1 correspondence between the safe dissociations of a
self-join-free conjunctive query and its query plans.  One simple
consequence is that \emph{every} query plan computes an upper
bound of the true probability.  For example, the two safe
dissociations above correspond to the plans
$P_2' = \pi_z(R \Join_x (\pi_x (S \Join_{x,y}  T)))$, and
$P_2''= \pi_z((\pi_{zy}(R \Join_x S)) \Join_{y} T)$.
We give an
intuitive system R-style algorithm~\cite{DBLP:conf/sigmod/SelingerACLP79} for enumerating all minimal safe dissociations of a
query $q$.  
Our algorithm takes into account important schema-level
information: functional dependencies 
and whether a relation is deterministic or probabilistic.  We prove
that our algorithm has several desirable properties that make it a
\emph{strict generalization of previous algorithms} described in the
literature: If $q$ is safe 
then the algorithm returns
only one safe plan that computes $q$ exactly; 
and if $q$ happens to
be safe on the particular database instance (e.g., the data happens
to satisfy a functional dependency), then one of the minimal safe dissociations
will compute the query exactly.  
(2) We 
use
relational optimization techniques 
to compute all minimal safe
dissociations of a query efficiently in the database engine.  Some
queries may have a large number of dissociations; e.g., a
8-chain query has 
4279 safe dissociations,
of which 429 are minimal.
Computing 429 queries sequentially in the database engine would still be
prohibitively expensive.  Instead, we tailor three relational query optimization techniques to dissociation:
($i$) combining all minimal plans into \emph{one single query}, 
($ii$) reusing \emph{common subexpressions} with views, and 
($iii$) performing \emph{deterministic semi-join reductions}.
(3) We conduct an experimental validation of our technique,
showing that, with all our optimizations enabled, computing hard
queries over probabilistic databases incurs only a modest penalty over
computing the same query on a deterministic database: For example, the
8-chain query runs only a factor of $<10$ slower than on a deterministic
database.  We also show that the dissociation-based technique has high
precision for ranking query answers based on their output
probabilities.

In summary, our three main contributions are:

\begin{enumerate}[nolistsep,label=(\arabic*)]
\item 
We describe an efficient algorithm for finding all minimal safe dissociations for self-join-free conjunctive queries in the presence of schema knowledge.  If the
  query is safe, then our algorithm returns a single minimal plan,
  which is the safe plan for the query (\autoref{sec:dissociation}).

\item We show how to apply three traditional query optimization techniques to dramatically improve the performance of the dissociation (\autoref{sec:optimizations}).

\item We perform a detailed experimental validation of our approach, showing 
  both its effectiveness in terms of query performance, and the quality
  of returned rankings. Our experiments also include
 a novel comparison between deterministic and probabilistic ranking approaches (\autoref{sec:experiments}).

\end{enumerate}

\emph{All proofs for this submission} together with additional illustrating examples are available in our technical report on arXiv~\cite{arxivDissociation:2013}.

%% file: 2background.tex
\section{Background}\label{sec:background}

\introparagraph{Probabilistic Databases} We fix a relational vocabulary
$\sigma = (R_1, \ldots, R_m)$.  A probabilistic database $D$ is a
database plus a function $p(t) \in [0,1]$ associating a probability to
each tuple $t \in D$.  
A {\em possible world} is a subset of $D$ generated by independently
including each tuple $t$ in the world with probability $p(t)$.  Thus,
the database $D$ is \emph{tuple-independent}.  
We use bold notation (e.g., $\vec x$) to denote
sets or tuples.
A \emph{self-join-free
  conjunctive query} is a first-order formula $q(\vec y) = \exists
x_1\ldots \exists x_k.(a_1 \wedge \ldots \wedge a_m)$ where each atom $a_i$
represents a relation $R_i(\vec x_i)$\footnote{We assume w.l.o.g.\ that
  $\vec x_i$ is a tuple of only variables without constants.}, the
variables $x_1, \ldots, x_k$ are called {\em existential variables},
and $\vec y$ are called the \emph{head variables} (or free
  variables).  The term ``self-join-free'' means that the atoms refer
to distinct relational symbols. We assume therefore w.l.o.g.\ that
every relational symbol $R_1, \ldots, R_m$ occurs exactly once in the
query.  Unless otherwise stated, a \emph{query} in this paper denotes
a self-join-free conjunctive query.  As usual, we abbreviate the query
by $q(\vec y) \datarule a_1, \ldots, a_m$,
and write $\HVar(q) = \vec y$, $\EVar(q) = \set{x_1, \ldots, x_k}$
and $\Var(q) = \HVar(q) \cup \EVar(q)$ for the set of head variables,
existential variables, and all variables of $q$. If $\HVar(q) =
\emptyset$ then $q$ is called a {\em Boolean} query.  We also write
$\Var(a_i)$ for the variables in atom $a_i$ and $\at(x)$ for the set of atoms that
  contain variable $x$. 
 The \emph{active
  domain} of a variable $x_i$ is denoted $\ADom_{x_i}$,\footnote{Defined formally as
  $\ADom_{x_i} = \bigcup_{j: x_i \in \Var(R_j)} \pi_{x_i}(R_j)$.}
 and the active domain of the entire database is $\ADom =
\bigcup_i \ADom_{x_i}$.  The focus of probabilistic query evaluation is
to compute $\PP{q}$; i.e.\ the probability that the query is true
in a randomly chosen world.

\introparagraph{Safe queries, safe plans} It is known that the data complexity of
any query $q$ is either in PTIME or \#P-hard. The former are called
\emph{safe queries} and are characterized precisely by a syntactic
property called \emph{hierarchical
  queries}~\cite{DBLP:journals/vldb/DalviS07}. We briefly review these
results:

\begin{definition}[Hierarchical query]\label{def:hierq}
Query $q$ is called hierarchical iff for any 
$x, y \in \EVar(q)$, 
one of the following three conditions
  hold: 
	$\at(x) \subseteq \at(y)$, 
	$\at(x) \cap \at(y) = \emptyset$, or
  	$\at(x) \supseteq \at(y)$.
\end{definition}

\noindent
For example, the query $q_1\datarule R(x,y), S(y,z), T(y,z,u)$ is
hierarchical, while $q_2 \datarule R(x,y), S(y,z), T(z,u)$ is not, as
neither of the three conditions holds for the variables $y$ and $z$.  

\begin{theorem}[Dichotomy~\cite{DBLP:journals/vldb/DalviS07}]\label{th:dichotomy}
  If $q$ is hierarchical, then
  $\PP{q}$ can be computed in PTIME in the size of $D$.
  Otherwise, computing $\PP{q}$ is \#P-hard
  in the size of $D$.
\end{theorem}

We next give an equivalent, recursive characterization of hierarchical
queries, for which we need a few definitions.  
We write $\SepVar(q)$ for the \emph{separator variables} (or root variables); i.e.\ the set of existential variables that appear in every atom.
$q$ is disconnected if its atoms can be partitioned into two
non-empty sets that do not share any existential variables 
(e.g., $q \datarule R(x,y),S(z,u),T(u,v)$ is disconnected and has
two connected components: ``$R(x,y)$'' and ``$S(z,u),T(u,v)$'').
For every set of
variables $\vec x$, denote $q - {\vec x}$ the query obtained by
removing all variables $\vec x$ (and decreasing the arities of the
relation symbols that contain variables from $\vec x$).

\begin{lemma}[Hierarchical queries] \label{lemma:hierarchical}
  $q$ is hierarchical iff either: (1) $q$ has a single atom; (2) $q$ has $k
  \geq 2$ connected components all of which are hierarchical; or (3)
  $q$ has a separator variable $x$ and $q-x$ is hierarchical.
\end{lemma}

\begin{definition}[Query plan]\label{def:queryPlans}
  Let $R_1, \ldots, R_m$ be a relational vocabulary. A query
    plan $P$ is given by the grammar
  \begin{align*}
   P 	::=&\, R_i(\vec x)
   \,\,\,|\,\, 	\projdd{\vec x} P
   \,\,\,|\,\,\! 	\joind{}{P_1, \ldots, P_k} 
 \end{align*}
 where $R_i(\vec x)$ is a relational atom containing the variables
 $\vec x$ and constants, $\projdd{\vec x}$ is the \emph{project operator with
 duplicate elimination}, and $\joind{}{\ldots}$ is the \emph{natural join} in prefix notation, which we
 allow to be $k$-ary, for $k \geq 2$.  We require that joins and
 projections alternate in a plan.  We do not distinguish between join
 orders, i.e.\ $\joind{}{P_1, P_2}$ is the same as $\joind{}{P_2,P_1}$.
\end{definition}

\noindent
We write $\HVar(P)$ for the head variables of $P$ (defined as the variables
$\vec x$ of the top-most projection $\projdd{\vec x}$,
or the union of
  the top-most projections if the last operation is a join).  
Every plan $P$ represents a query $q_P$
defined by taking all atoms mentioned in $P$ and setting
$\HVar(q_P) = \HVar(P)$.
For notational convenience, we also use the ``project-away'' notation,
by writing $\projd{\vec y}(P)$ instead of $\projdd{\vec x}(P)$, where
$\vec y$ are the variables being projected away; i.e. $\vec y =
\HVar(P) - \vec x$.

Given a probabilistic database $D$ and a plan $P$, each output tuple
$t \in P(D)$ has a 
$\score(t)$, defined
inductively on the structure of $P$ as follows: 
If $t \in R_i(\vec
x)$, then $\score(t)= p(t)$, i.e.\ its probability in $D$;
if $t \in \,
\joind{}{P_1(D), \ldots, P_k(D)}$ where $t = \joind{}{t_1, \ldots,
  t_k}$, then $\score(t) = \prod_{i=1}^k\score(t_i)$; and
if $t \in
\projdd{\vec x}(P(D))$, and $t_1, \ldots, t_n \in P(D)$ are all the
tuples that project into $t$, then $\score(t) = 1 - \prod_{i=1}^n
(1-\score(t_i))$.  
In other words, $\score$ computes a probability by
assuming that all tuples joined by $\Join$ are independent, and
all duplicates eliminated by $\projdd{}$ are also independent. If
these conditions hold, then $\score$ is the correct query probability,
but in general the score is different from the probability. Therefore, 
\emph{$\score$ is not equal to the probability, in general}, and is also called an
extensional semantics~\cite{DBLP:journals/tois/FuhrR97,Pearl:1988pb}.
For a Boolean plan $P$, we get one single score,
which we denote $\textit{score}(P)$.  

The requirement that joins and projections alternate is w.l.o.g.\
because nested joins like $\joind{}{\joind{}{R_1, R_2}, R_3}$ can
be rewritten into $\joind{}{R_1, R_2, R_3}$ while keeping the
same probability score.  For the same reason we do not distinguish
between different join orders.

\begin{definition}[Safe plan]\label{def:safePlan}
  A plan $P$ is called safe iff, for any join operator
  $\joinp{}{P_1, \ldots, P_k}$, all subplans have the same head variables:
  $\HVar(P_i) = \HVar(P_j)$ for all $1 \leq i, j \leq k$.
\end{definition}

The recursive definition of \autoref{lemma:hierarchical} gives us
immediately a safe plan for a hierarchical query.  Conversely, every
safe plan defines a hierarchical query.  The following summarizes our
discussion:

\begin{proposition}[Safety \cite{DBLP:journals/vldb/DalviS07}]\label{prop:uniqueSafePlan}
  (1) Let $P$ be a plan for the query $q$. Then $\textit{score}(P) =
  \PP{q}$ for any probabilistic database iff $P$ is safe. (2)
  Assuming \#P$\neq$PTIME, a query $q$ is safe (i.e.\ $\PP{q}$ has
  PTIME data complexity) iff it has a safe plan $P$; in that case the
  safe plan is unique, and $\PP{q} = \score(P)$.

\end{proposition}

\introparagraph{Boolean Formulas} Consider a set of Boolean variables $\mathbf{X}
= \set{X_1, X_2, \ldots}$ and a probability function $p : \mathbf{X}
\rightarrow [0,1]$.  Given a Boolean formula $F$, denote $\PP{F}$ the
probability that $F$ is true if each variable $X_i$ is independently true
with probability $p(X_i)$.  In general, computing
$\PP{F}$ is \#P-hard in the number of variables $\mathbf{X}$.  If $D$
is a probabilistic database then we interpret every tuple $t \in D$ as
a Boolean variable and denote the lineage of a Boolean $q \datarule
g_1,\ldots,g_m$ on $D$ as the Boolean DNF formula $F_{q,D} =
\bigvee_{\theta: \theta \models q}\theta(g_1)\wedge\cdots\wedge
\theta(g_m)$, where $\theta$ 
ranges over all assignments of $\EVar(q)$ that satisfy $q$ on $D$.
It is well known that $\PP{q} = \PP{F_{q,D}}$. In other
words the probability of a Boolean query is the same as the
probability of its lineage formula.

\begin{example}[Lineage]\label{ex:simple} 
	If $F= XY \vee XZ$ then $\PP{F} =
  p(1\!-\!(1\!-\!q)(1\!-\!r))
 = pq + pr - pqr$, 
where $p = p(X), q = p(Y)$, and $r =
  p(Z)$.  Consider now the query $q \datarule R(x), S(x,y)$ over the
database
$D = \set{R(1),R(2),S(1,4),S(1,5)}$.  Then the
  lineage formula is $F_{q,D} = R(1)\wedge S(1,4) \vee R(1) \wedge
  S(1,5)$, i.e.\ same as $F$, up to variable renaming. It is now easy to
  see that $\PP{q} = \PP{F_{q,D}}$.
\end{example}

A key technique that we use in this paper is the following result
from~\cite{DBLP:journals/tods/GatterbauerS14}:  Let $F, F'$ be two
Boolean formulas with sets of variables $\mathbf{X}$ and $\mathbf{X}'$,
respectively.  We say that $F'$ is a {\em dissociation} of $F$ if
there exists a substitution $\theta: \mathbf{X}' \rightarrow
\mathbf{X}$ such that $F'[\theta]=F$.  If $\theta^{-1}(X) =
\set{X',X'', \ldots}$ then we say that the variable $X$ {\em
  dissociates into} $X', X'', \ldots$; if $|\theta^{-1}(X)| = 1$ then
we assume w.l.o.g.\ that $\theta^{-1}(X) = X$ (up to variable renaming)
and we say that $X$ does not dissociate.  Given a probability function
$p : \mathbf{X} \rightarrow [0,1]$, we extend it to a probability
function $p' : \mathbf{X}' \rightarrow [0,1]$ by setting $p'(X') =
p(\theta(X'))$.  Then, we have shown:

\begin{theorem}[Oblivious DNF bounds~\cite{DBLP:journals/tods/GatterbauerS14}] \label{th:bool:dissoc} Let
  $F'$ be a monotone DNF formula that is a dissociation of $F$ through the
  substitution $\theta$.  Assume that for any variable $X$, no two
  distinct dissociations $X', X''$ of $X$ occur in the same prime
  implicant of $F'$.  Then (1) $\PP{F} \leq \PP{F'}$, and (2) if all
  dissociated variables $X \in \mathbf{X}$ are deterministic (meaning:
  $p(X)=0$ or $p(X)=1$) then $\PP{F}=\PP{F'}$.
\end{theorem}

Intuitively, a dissociation $F'$ is obtained from a formula $F$ by
taking different occurrences of a variable $X$ and replacing them with
fresh variables $X', X'', \ldots$; in doing this, the probability of
$F'$ may be easier to compute, giving us an upper bound for $\PP{F}$.

\begin{example}[\autoref{ex:simple} cont.]
$F' = X'Y \vee X''Z$ is a dissociation of $F = XY \vee XZ$, and its probability is $\PP{F'} =
1\!-\!(1\!-\!pq)(1\!-\!pr) = pq + pr - p^2qr$.  Here, only the variable $X$
dissociates into $X', X''$.  It is easy to see that $\PP{F} \leq
\PP{F'}$. Moreover, if $p = 0$ or 1, then $\PP{F}=\PP{F'}$.
The condition that no two dissociations of the same variable occur in
a common prime implicant is necessary: for example, $F'= X'X''$ is a
dissociation of $F = X$. However $\PP{F} = p$, $\PP{F'} = p^2$,  and
we do not have $\PP{F} \leq \PP{F'}$.
\end{example}

%% file: 3-1queryDissociation.tex
\section{Dissociation of queries}\label{sec:dissociation}
\noindent
This section introduces our main technique 
for approximate query processing.
After defining dissociations (\autoref{sec:queryDissociationIntro}), we show
that some of them are in 1-to-1 correspondence with query plans, then derive our
first algorithm for approximate query processing (\autoref{sec:plans}).  Finally, we
describe two extensions
in the presence of 
deterministic relations or functional
dependencies (\autoref{sec:schemaKnowledge}).

\subsection{Query dissociation}\label{sec:queryDissociationIntro}

\begin{definition}[Dissociation]\label{def:Dissociation}
	Given a Boolean query $q \datarule R_1(\vec x_1), \ldots, R_m(\vec x_m)$
	and a probabilistic database $D$.
  Let $\Delta = (\vec y_1, \ldots, \vec y_m)$ be a collection of sets
  of variables with $\vec y_i \subseteq \Var(q)- \Var(g_i)$ for every relation $R_i$.
The
  \emph{dissociation} defined by $\Delta$ has then two components:
\begin{enumerate}[nolistsep,label=(\arabic*)]

\item the \emph{dissociated query}:
$q^{\Delta} \datarule R_1^{\vec y_1}(\vec x_1, \vec y_1), \ldots, R_m^{\vec y_m}(\vec x_m, \vec y_m)$,
where each $R_i^{\vec y_i}(\vec x_i, \vec
y_i)$ is a new relation of arity $|\vec x_i| + |\vec y_i|$.

\item the \emph{dissociated database instance} $D^\Delta$ consisting of
  the tables over the vocabulary $\sigma^\Delta$ obtained by
  evaluating (deterministically) the following queries over the
  instance $D$:
  \begin{align*}
    R_i^{\vec y_i}(\vec x_i, \vec y_i) & \datarule R_i(\vec x_i),  \ADom_{y_{i1}}(y_{i1}), \ldots,
  	\ADom_{y_{ik}}(y_{ik})
  \end{align*}
  where $\vec y_i = (y_{i1}, \ldots, y_{ik_i})$. For each tuple $t' \in R_i^{\vec y_i}$, its probability is defined
  as $p'(t') = p(\pi_{\vec x_i}(t'))$, i.e.\ 
  the probability of $t$ in the database $D$.
\end{enumerate}
\end{definition}

Thus, a dissociation acts on both the query expression and the
database instance: It adds some variables $\vec y_i$ to each
relational symbol $R_i$ of the query expression, and it computes a new
instance for each relation $R_i^{\vec y_i}$ by copying every record
$t \in R_i$ once for every tuple in the cartesian product
$\ADom_{y_{i1}} \times \cdots \times \ADom_{y_{ik}}$.  When $\vec y_i =
\emptyset$ then we abbreviate $R_i^\emptyset$ with $R_i$.  
We give a simple example:

\begin{example}[\autoref{ex:simple} cont.]
Consider $q \datarule R(x),S(x,y)$. Then $\Delta = (\set{y}, \emptyset)$ defines the
following dissociation: $q^\Delta = R^{y}(x,y),S(x,y)$, and the new
relation $R^y$ contains the tuples $R^y(1,4), R^y(1,5), R^y(2,4),
R^y(2,5)$.  Notice that the lineage of the dissociated query $q^\Delta$ is
$F_{q^{\Delta},D^{\Delta}} = R^y(1,4),S(1,4) \vee R^y(1,5),S(1,5)$ and
is the same (up to variable renaming) as the dissociation of the
lineage of query $q$: $F'= X'Y \vee X''Z$.  
\end{example}

\begin{theorem}[Upper query bounds]\label{th:upperBounds}
  For every dissociation $\Delta$ of $q$: $\PPP(q) \leq \PPP(q^{\Delta})$.
\end{theorem}

\begin{proof}
  \Autoref{th:upperBounds} follows immediately from \autoref{th:bool:dissoc} by
  noting that the lineage $F_{q^\Delta,D^\Delta}$ is a dissociation of
  the lineage $F_{q,D}$ through the substitution $\theta : D^\Delta
  \rightarrow D$ defined as follows: for every tuple $t' \in R_i^{\vec
    y_i}$, $\theta(t') = \pi_{\vec x_i}(t')$.
\end{proof}

\begin{definition}[Safe dissociation]
  A dissociation $\Delta$ of a query $q$ is called \emph{safe} 
if the dissociated query ${q^{\Delta}}$ is safe.
\end{definition}

By \autoref{th:dichotomy}, a dissociation is safe  (i.e.\ its probability can be evaluated in PTIME) iff $q^\Delta$ is hierarchical.
Hence, amongst all dissociations, we are interested in those that are easy to evaluate
and use them as a technique to approximate the probabilities of queries
that are
hard to compute. The idea is simple: 
Find a safe dissociation $\Delta$, compute $\PP{q^\Delta}$, and thereby obtain an upper bound on $\PP{q}$.  In fact,
we will consider \emph{all} safe dissociations and take the
minimum of their probabilities, since this gives an even better upper bound
on $\PP{q}$ than that given by a single dissociation.  We call this
quantity the \emph{propagation score}\footnote{We chose the name ``propagation'' for our method because of 
  similarities with efficient belief propagation algorithms in
  graphical models. See~\cite{arxivDissociation:2013} for a
  discussion on how query dissociation generalizes relevance propagation
 from graphs to hypergraphs, and \cite{LinBP:arxiv} for a recent approach for speeding up  belief propagation even further.} of the query $q$.

\begin{definition}[Propagation]
  The \emph{propagation score $\rho(q)$ for a query $q$} is the minimum score of all safe dissociations: $\rho(q) = \min_{\Delta} \PPP(q^{\Delta})$ with $\Delta$ ranging over all safe dissociations.
\end{definition}

The difficulty in computing $\rho(q)$ is the number of
dissociations that is large even for relatively small queries: If $q$ has
$k$ existential variables and $m$ atoms, then $q$ has $2^{|K|}$
possible dissociations with $K = \sum_{i=1}^m \big( k- |
\texttt{Var}(g_i) | \big)$ forming a partial order in the shape of a
power set lattice 
(see \autoref{Fig_PartialDissociationOrderExample_a}).  
Therefore, our next step is to
prune the space of dissociations and to examine only the 
minimum number necessary.
We start by defining a
partial order on dissociations:

\begin{definition}[Partial dissociation order]
	We define the partial order on the dissociations of a query as:
	$$
		\Delta \preceq \Delta'   \,\,\Leftrightarrow\,\,  \forall i: \vec y_i \subseteq \vec y_i'
	$$
\end{definition}

Whenever $\Delta \preceq \Delta'$, then $q^{\Delta'}, D^{\Delta'}$ is
a dissociation of $q^{\Delta}, D^{\Delta}$ (given by $\Delta'' =
\Delta' - \Delta$).  Therefore, we obtain immediately:

\begin{corollary}[Partial dissociation order]\label{cor:partialDissociationOrder}
  If $\Delta \preceq \Delta'$ then $\PPP(q^\Delta) \leq
  \PPP(q^{\Delta'})$.
\end{corollary}

\begin{example}[Partial dissociation order]\label{ex:partialDissociation}
	Consider the query $q \datarule R(x), S(x), T(x,y), U(y)$.
	It is unsafe and allows $2^3 = 8$ dissociations which are shown in \autoref{Fig_PartialDissociationOrderExample_a} with the help of an ``augmented incidence matrix'':
each row represents one relation and each column one variable:
An empty circle ($\circ$) indicates that a relation contains a variable;
a full circle ($\bullet$) indicates that a relation is dissociated on a variable (the reason for using two separate symbols becomes clear when we later include domain knowledge).
	Among those 8 dissociations, 5 are safe, shaded in green, and have the hierarchy among variables highlighted. Furthermore, 2 of the 5 safe dissociations are minimal:
$q^{\Delta_3}	 \datarule  R(x), S(x), T(x,y), U^x(x,y) $, and
$q^{\Delta_4}	 \datarule  R^y(x, y), S^y(x, y), T(x,y), U(y)$ .
To illustrate that these dissociations are upper bounds, consider a database with 
$R=T=U =\{1,2\}$, 
$S=\{(1,1),(1,2),(2,2)\}$, 
and the probability of all tuples $=\frac{1}{2}$. 
Then $q$ has probability  $\frac{83}{2^9} \approx 0.161$, 
while $q^{\Delta_3}$ has probability $\frac{169}{2^{10}} \approx 0.165$,
and $q^{\Delta_4}$ has probability $\frac{353}{2^{11}} \approx 0.172$, both of which are upper bounds.	
The propagation score is the minimum score of all minimal safe dissociations and thus $\approx 0.165$.
\markend
\end{example}

\begin{figure}[t]
	\vspace{-3mm}
    \centering
	\subfloat[]{
		\label{Fig_PartialDissociationOrderExample_a}
		\hspace{-4mm}
		\includegraphics[scale=0.43]{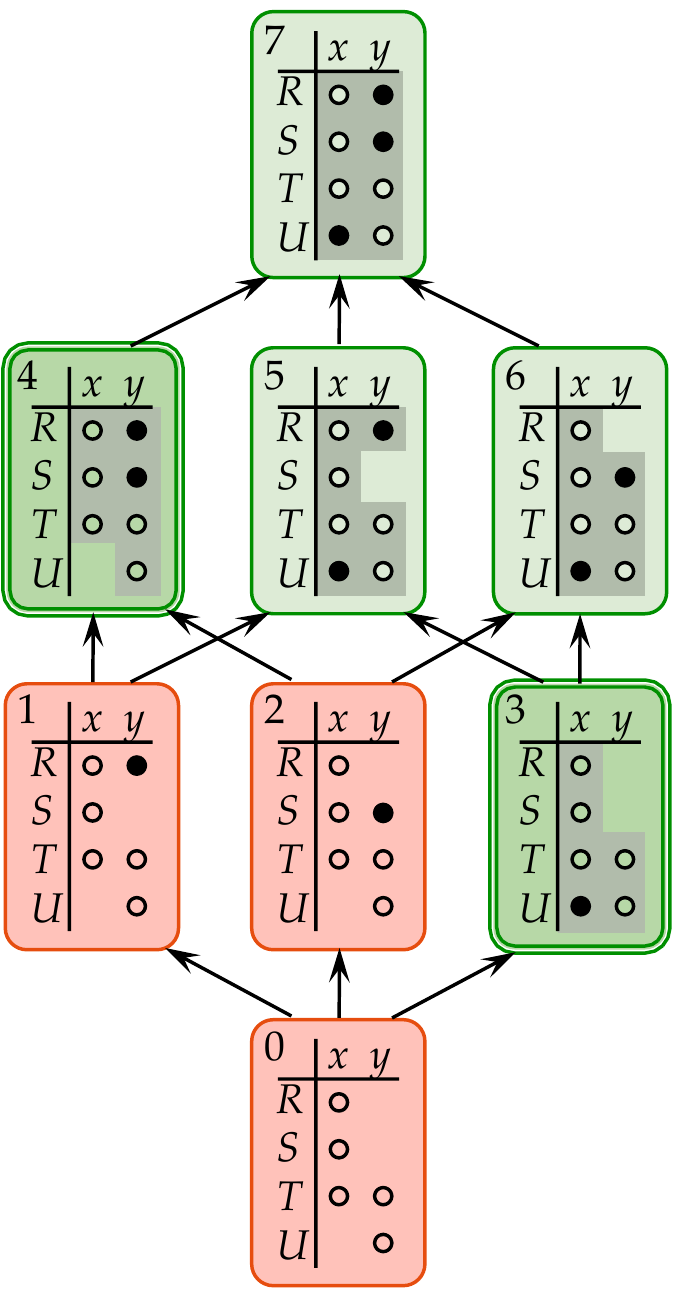}
		\hspace{-3mm}}
	\hspace{4mm}	
	\subfloat[]{
		\includegraphics[scale=0.53]{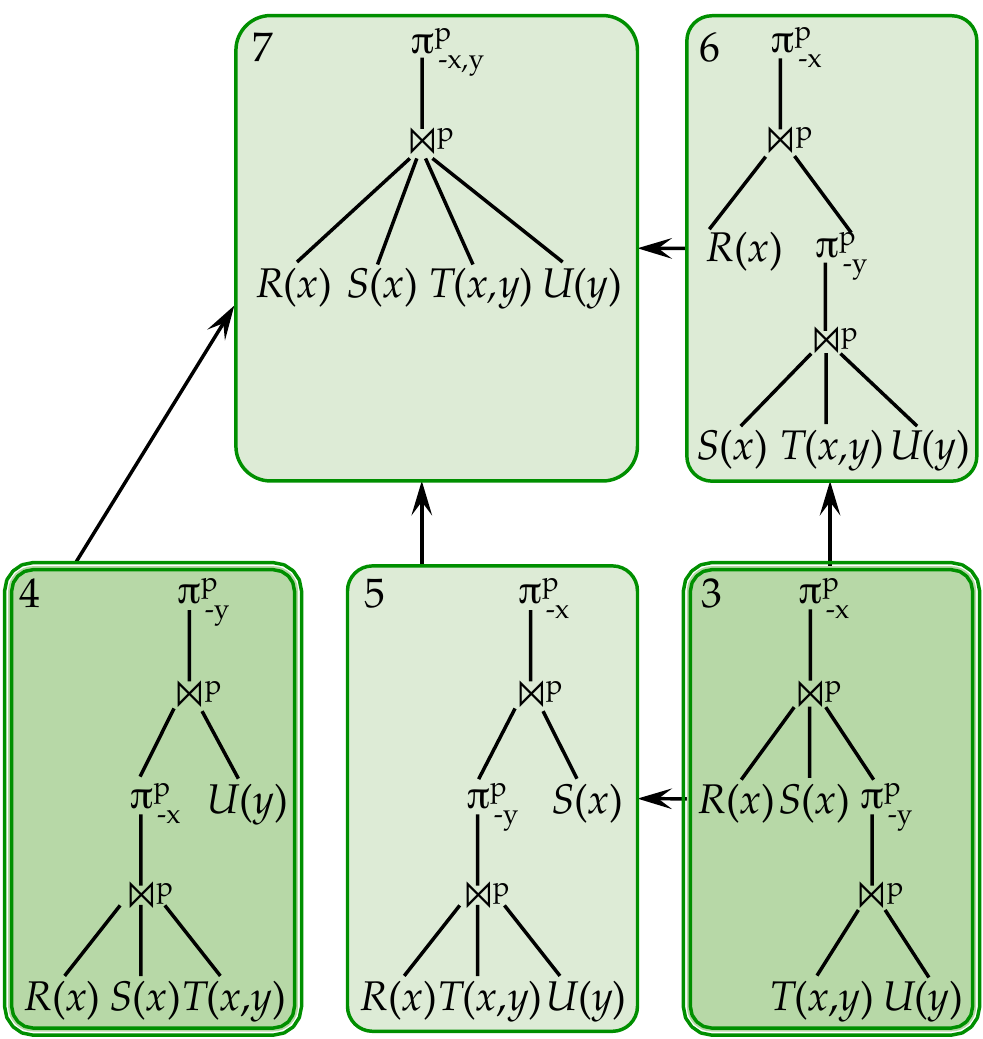}
		\hspace{-2mm}
		\label{Fig_PartialDissociationOrderExample_b}}
	\hspace{0mm}
	\caption{
		\autoref{ex:partialDissociation} 
		(a):
		Partial dissociation order for $q \datarule R(x),S(x),$ $T(x,y), U(y)$.
		Safe dissociations are green and have the hierarchies between variables shown (3 to 7), 
		\emph{minimal safe dissociations} are dark green and double-lined (3 and 4). 
		(b): 
		All 5 query plans for $q$ 
		and their correspondence to safe dissociations
		 (3 to 7).
		}
	\label{Fig_PartialDissociationOrderExample}
	\vspace{-1mm}
\end{figure}

In general, the set of dissociations forms a lattice, with the
smallest element $\Delta_\bot = (\emptyset, \ldots, \emptyset)$
($q^{\Delta_\bot}=q$) and the largest element $\Delta_\top =
(\Var(q)-\Var(g_1), \ldots, \Var(q)-\Var(g_m))$ ($q^{\Delta_\top}$ is
safe, since every atom contains all variables).  As we move up in the
lattice the probability increases, but the safe/unsafe status may
toggle arbitrarily from safe to unsafe and back.  For example $q \datarule
R(x),S(x),T(y)$ is safe, its dissociation $q' \datarule R(x),S^y(x,y),T(y)$ is
unsafe, yet the next dissociation $q'' \datarule R(x),S^y(x,y),T^x(x,y)$ is safe
again.

This suggests the following naive algorithm for computing $\rho(q)$:
Enumerate all dissociations $\Delta_1, \Delta_2, \ldots$ by traversing
the lattice breadth-first, bottom up (i.e.\ whenever $\Delta_i \prec
\Delta_j$ then $i < j$).  For each dissociation $\Delta_i$, check if
$q^{\Delta_i}$ is safe. If so, then first update $\rho \leftarrow \min(\rho,
\PPP(q^{\Delta_i}))$, then remove from the list all dissociations
$\Delta_j \succ \Delta_i$.  However, this algorithm is inefficient for
practical purposes for two reasons: ($i$) we need to iterate over
many dissociations in order to discover those that are safe; and
($ii$) computing
$\PPP(q^{\Delta_i})$ requires computing a new database instance
$D^{\Delta_i}$ for each safe dissociation $\Delta_i$.
We show in the next section how to avoid both sources
of inefficiency by 
exploiting the lattice structure and by
iterating over query plans instead of safe
dissociations.

%% file: 3-2queryDissociationPlans.tex
\subsection{Dissociations and Plans}\label{sec:plans}

\noindent
We prove here that the safe dissociations $q^\Delta$ are in 1-to-1
correspondence with query plans of the original query $q$. This allows us
to ($i$) efficiently find safe dissociations (by iterating over
query plans instead of all dissociations), and to ($ii$) compute $\PPP(q^\Delta)$ without
having to materialize the dissociated database $D^\Delta$.

We next describe the 1-to-1 mapping.  Consider a safe dissociation
$q^{\Delta}$ and denote its corresponding unique safe plan
$P^\Delta$. This plan uses dissociated relations, hence each relation
$R^{\vec y_i}_i(\vec x_i, \vec y_i)$ has extraneous variables
$\vec y_i$.  Drop all variables $\vec y_i$ from the relations and
all operators using them: This transforms $P^\Delta$ into a
regular, generally unsafe plan $P$ for $q$.  For a trivial example,
the plan corresponding to the top dissociation $\Delta_\top$ of a
query $q$ is $\projd{\Var(q)} (\joind{}{P_1, \ldots, P_k})$: It
performs all joins first, followed by all projections.

Conversely, consider any plan $P$ for $q$.  We define its
corresponding safe dissociation $\Delta^P$ as follows.  For each join
operation $\joinp{}{P_1,\ldots,P_k}$, let its \emph{join variables}
$\JVar$ be the union of the head variables of all subplans: $\JVar =
\bigcup_j \HVar(P_j)$. For every relation $R_i$ occurring in $P_j$,
add the missing variables $\JVar - \HVar(P_j)$ to $\vec y_i$.  For
example, consider $\joinp{}{R(x), T(x,y), U(y)}$ (this is the lower
join in query plan 5 of
\autoref{Fig_PartialDissociationOrderExample_b}).  Here, $\JVar =
\set{x,y}$, and the corresponding safe dissociation of this subplan is
$q^\Delta(x,y) \datarule R^y(x, y),T(x,y), U^x(x,y)$.  Note that while
there is a one-to-one mapping between safe dissociations and query
plans, unsafe dissociations do not correspond to plans.

\begin{theorem}[Safe dissociation]\label{th:safeDissociation}
  Let $q$ be a conjunctive query without self-joins. (1) The mappings
  $\Delta \mapsto P^\Delta$ and $P \mapsto \Delta^P$ are inverses of
  each other. 
(2) For every safe dissociation
  $\Delta$, $\PPP(q^{\Delta}) = \score(P^{\Delta})$.
\end{theorem}

\begin{corollary}[Upper bounds]\label{cor:queryPlanBounds}
  Let $P$ be any plan for a Boolean query $q$.  Then $\PPP(q) \leq
  \score(P)$.
\end{corollary}

The proof follows immediately from $\PPP(q) \leq \PPP(q^{\Delta^P})$
(\autoref{th:upperBounds}) and $\PPP(q^{\Delta^P}) = \score(P)$
(\autoref{th:safeDissociation}).  In other words, \emph{any} plan for
$q$ computes a probability score that is guaranteed to be an upper
bound on the correct probability $\PPP(q)$.

\noindent

\autoref{th:safeDissociation} suggests the following improved
algorithm for computing the propagation score $\rho(q)$ of a query:
Iterate over all plans $P$, compute their scores, and retain the
minimum score $\min_P [\score(P)]$. Each plan $P$ is evaluated
directly on the original probabilistic database, and there no need to
materialize the dissociated database instance. However, this approach
is still inefficient because it computes several plans that
correspond to non-minimal dissociations.  For example, in
\autoref{Fig_PartialDissociationOrderExample} plans 5, 6, 7 correspond
to non-minimal dissociations, since plan 3 is safe and below them.

\introparagraph{Enumerating minimal safe dissociations}
Call a plan $P$ \emph{minimal} if $\Delta^P$ is minimal in the set of
safe dissociations.  For example, in \autoref{ex:partialDissociation},
the minimal plans are 3 and 4.  The propagation score is thus the
minimum of the scores of the two minimal plans: $\rho(q) = \min_{i \in
  \{3,4\}} \big[ \score\big(P^{(i)}\big) \big]$.  Our improved
algorithm will iterate only over minimal plans, by relying on a
connection between plans and sets of variables that disconnect a query:
A \emph{cut-set} is a set of existential
variables $\vec x \in \EVar(q)$ s.t.\ $q - {\vec x}$ is disconnected. 
A \emph{min-cut-set} (for minimal cut-set) is a cut-set for which no strict subset is a cut-set.
We
denote $\MinCuts(q)$ the set of all min-cut-sets.  
Note that $q$ is
disconnected iff $\MinCuts(q) = \set{\emptyset}$.

The connection between $\MinCuts(q)$ and query plans is given by two
observations: 
(1)  Let $P$ be any plan for $q$. If $q$ is connected, then
the last operator in $P$ is a projection, i.e.\  $P = \projd{\vec x}
(\joind{}{P_1, \ldots, P_k})$, and the projection variables $\vec x$
are the join variables ${\vec x} = \JVar$
because $q$ is Boolean so the plan must project away all
variables.  We claim that $\vec x$ is a cut-set
for $q$ and that $q- {\vec x}$ has $k$ connected
components corresponding to $P_1, \ldots, P_k$.  
Indeed, if $P_i, P_j$
share any common variable $y$, then they must join on $y$, hence $y
\in \JVar$. Thus, 
\emph{cut-sets are
in 1-to-1 correspondence with the top-most projection operator} of a
plan.  
(2) Now suppose that $P$ corresponds to a safe dissociation
$\Delta^P$, and let $P' = \projd{\vec x} (\joind{}{P_1', \ldots,
  P_k'})$ be its unique safe plan.  
Then 
${\vec x} = \SepVar(q^{\Delta^P})$; i.e.\ the top-most project
operator removes all separator variables.\footnote{This
  follows from the recursive definition of the unique safe plan of a query in \autoref{lemma:hierarchical}: the top
  most projection consists precisely of its separator variables.}  
Furthermore, if $\Delta
\succeq \Delta^P$ is a larger dissociation, then
$\SepVar(q^{\Delta}) \supseteq \SepVar(q^{\Delta^P}) $ 
(because any
separator variable of a query continues to be a separator variable in
any dissociation of that query).  Thus, 
\emph{minimal plans correspond to min-cut-sets}; in other words,
$\MinCuts(q)$ is in 1-to-1 correspondence with the top-most projection
operator of {\em minimal} plans.

Our discussion leads immediately to \autoref{alg:basicAlgorithm} for computing the
propagation score $\rho(q)$. 
It also applies to non-Boolean queries
by treating the head variables as constants, hence 
ignoring them when
computing connected components.
The algorithm proceeds recursively.  If $q$ is a single
atom then it is safe and we return its unique safe plan.  If the query
has more than one atom, then we consider two cases, when
$q-\HVar(q)$ is disconnected or connected.  In the first case, every
minimal plan is a join, where the subplans are minimal plans of the
connected components.  In the second case, a minimal plan results from a
projection over min-cut-sets.
Notice that recursive calls of the algorithm will alternate between
these two cases, until they reach a single atom.

\begin{algorithm2e}[t]
\scriptsize
\caption{generates all minimal query plans for a given query $q$.}\label{alg:basicAlgorithm}
\SetKwInput{Algorithm}{Recursive algorithm}
\SetKwFunction{MP}{MP}
\SetKwFor{ForAll}{forall}{do}{endfch}	
\Algorithm{\FuncSty{MP (EnumerateMinimalPlans)}}
\KwIn{Query $q(\vec x) \datarule R_1(\vec x_1), \ldots, R_m(\vec x_m)$}	
\KwOut{Set of all minimal query plans $\mathcal{P}$}
\BlankLine
\lIf{$m=1$}{$\mathcal{P} \leftarrow \{\projdd{\vec x} R_1(\vec x_1)\}$}\label{alg1:line1} 
\Else{
	Set $\mathcal{P} \leftarrow \emptyset$ \;
	\uIf{$q$ is disconnected}{\label{alg1:disconnected}
		Let $q = q_1, \ldots, q_k$ be the connected components  of $q-\HVar(q)$\;
		\lForEach{$q_i$}{Let $\HVar(q_i) \leftarrow \HVar(q) \cap \Var(q_i)$} 
		\ForEach{$(P_1, \ldots, P_k) \in \MP(q_1) \times \cdots \times \MP(q_k)$} {
			$\mathcal{P} \leftarrow \mathcal{P} \cup \{\joinp{}{P_1, \ldots, P_k}\}$ \;
		}
	}
	\label{alg1:connected}\Else{
		\ForEach{$\vec y \in \MinCuts(q-\HVar(q))$ \label{alg1:line10}}{
			Let $q' \leftarrow q$ with $\HVar(q') \leftarrow \HVar(q) \cup \vec y$ \;			
			\lForEach{$P \in \MP(q')$}{
				$\mathcal{P} \leftarrow \mathcal{P} \cup \{\projd{\vec y} \, P\}$ \label{alg1:line12}
			}
		}	
	}	
}
\end{algorithm2e}

\begin{theorem}[\Autoref{alg:basicAlgorithm}]\label{prop:algorithmSound}
  \Autoref{alg:basicAlgorithm} computes the set of all minimal query
  plans.
\end{theorem}

\introparagraph{Conservativity} Some probabilistic database systems first check
if a query $q$ is safe, and in that case compute the exact probability
using the safe plan, otherwise use some approximation technique.  
We show that \autoref{alg:basicAlgorithm}
is conservative, in the sense that, if $q$ is safe,
then $\rho(q) = \PP{q}$.
Indeed, in that
case $\MP(q)$ returns a single plan, namely the safe $P$ for $q$,
because the empty dissociation, $\Delta_\bot = (\emptyset, \ldots,
\emptyset)$, is safe, and it is the bottom of the dissociation
lattice, making it the unique minimal safe dissociation.

\introparagraph{Score Quality} 
We show
here that the approximation of $\PP{q}$ by $\rho(q)$ becomes tighter
as the input probabilities in $D$ decrease. Thus, the smaller the probabilities in the database, the closer does the
ranking based on the propagation score approximate the ranking by the actual probabilities.

\begin{proposition}[Small probabilities]\label{prop:smallProbabilities}
Given a query $q$ and database $D$.
Consider the operation of scaling down the probabilities of all tuples in $D$ with a factor $f<1$. Then the relative error of approximation of $\PPP(q)$ by the propagation score $\rho(q)$ decreases as $f$ goes to 0:
$\lim_{f \rightarrow 0} \frac{\rho(q) - \PPP(q)}{\PPP(q)} \rightarrow 0$.
\end{proposition}

\introparagraph{Number of Dissociations} While the number of minimal safe dissociations
is exponential in the size of the query, recall that it is 
independent of the size of the database.
\Autoref{table:numberMinimalQueryPlans} gives an overview of the
number of minimal query plans, total query plans, and all
dissociations for $k$-star and $k$-chain queries (which are later used
in \autoref{sec:experiments}).
Later \autoref{sec:optimizations} gives optimizations that allow us to evaluate a large number of plans efficiently.

\begin{figure}[t]
\centering	
\small
\renewcommand{\tabcolsep}{0.9mm}
\renewcommand{\arraystretch}{0.95}	
	\begin{tabular}[t]{@{\hspace{1pt}} >{$}c<{$} | >{$}r<{$} >{$}r<{$} >{$}r<{$} || 
		>{$}r<{$} | >{$}r<{$} >{$}r<{$} >{$}r<{$} @{\hspace{1pt}}}
		 	\multicolumn{4}{c || }{$k$-star query}	&\multicolumn{4}{c }{$k$-chain query}\\
		k 	&\multicolumn{1}{c}{\#MP}	&\multicolumn{1}{c}{\#P}	&\multicolumn{1}{c||}{$\#\Delta$}	
		&k	&\multicolumn{1}{c}{\#MP}	&\multicolumn{1}{c}{\#P}	&\multicolumn{1}{c}{$\#\Delta$}	\\
		\hline
		1	&1 		&1		 &1			&2  &1  		&1		&1		   \\ 
		2	&2 		&3		 &4			&3  &2  		&3		&4		   \\ 
		3	&6 		&13		 &64		&4  &5  		&11		&64	       \\ 
		4	&24 	&75		 &4096		&5  &14  		&45		&4096	   \\ 
		5	&120	&541	 &>10^6		&6  &42  		&197	&>10^6	   \\ 
		6	&720	&4683	 &>10^9		&7  &132  		&903	&>10^9	   \\ 
		7	&5040	&47293	 &>10^{12}	&8  &429  		&4279	&>10^{12}   \\ 					
		\hline 
		\textrm{seq}	&k!		&\href{http://oeis.org/classic/A000670}{A000670} 	 &2^{k(k\!-\!1)}
		& \textrm{seq}			
			&\href{http://oeis.org/classic/A000108}{A000108}  
			&\href{http://oeis.org/classic/A001003}{A001003}	&2^{(k\!+\!1)k}	   \\ 		
	\end{tabular}
\caption{Number of minimal plans, total plans, and total dissociations for star and chain queries (A are  OEIS sequence numbers~\protect\cite{oeis}).}
\label{table:numberMinimalQueryPlans}
\end{figure}

%% file: 3-3schemaKnowledge.tex
\subsection{Minimal plans with schema knowledge}\label{sec:schemaKnowledge}

\noindent Next, we show how knowledge of \emph{deterministic relations} (i.e.\ all tuples have
probability $= 1$), and \emph{functional dependencies} can reduce the number of plans needed to
calculate the propagation score.

\subsubsection{Deterministic relations (DRs)}\label{sec:deterministicOptimization}

\noindent Notice that we can treat deterministic relations (DRs) just like probabilistic
relations, and \autoref{cor:queryPlanBounds} with $\PP{q} \leq score(P)$ still holds for any
plan $P$. Just as before, our goal is to find a minimum number of plans that compute the
minimal score of \emph{all plans}: $\rho(q) = min_P score(P)$. It is known that an unsafe query
$q$ can become safe (i.e., $\PP{q}$ can be calculated in PTIME with one single plan) if we
consider DRs. Thus, in particular, we would still like an improved algorithm that returns one
single plan if a query with DRs is safe. The following lemma will help us achieve this goal:

\begin{lemma}[Dissociation and DRs]\label{lemma:DetDissociation}
Dissociating a deterministic relation does not change the probability.
\end{lemma}

\begin{proof}
\Autoref{lemma:DetDissociation} follows immediately from \specialref{Theorem}{th:bool:dissoc}{(2)} and noting that dissociating tuples in DRs corresponds exactly to dissociating variables $\vec X$ with $p{(X_i)} = 1$.
\end{proof}

We thus define a new \emph{probabilistic dissociation preorder $\preceq^p$} by:
\begin{align*}
  \Delta \preceq^p \Delta' \Leftrightarrow \forall i, R_i \mbox{ probabilistic}: {\vec y}_i \subseteq {\vec y}_i'
\end{align*}

\noindent In other words, $\Delta \preceq^p \Delta'$ still implies $\PPP(q^\Delta) \leq
\PPP(q^{\Delta'})$, but $\preceq^p$ is defined on probabilistic relations only. Notice, that
for queries without DRs, the relations $\preceq^p$ and $\preceq$ coincide. However, for queries
with DRs, $\preceq^p$ is a preorder, not an order. Therefore, there exist distinct
dissociations $\Delta$, $\Delta'$ that are equivalent under $\preceq^p$ (written as $\Delta
\equiv^p \Delta'$), and thus have the same probability: $\PPP(q^\Delta) = \PPP(q^{\Delta'})$.
As a consequence, using $\preceq^p$ instead of $\preceq$, allows us to further reduce the
number of minimal safe dissociations.

\begin{example}[DRs]\label{ex:DRsSimple}
Consider $q \datarule R(x),S(x,y),T^d(y)$ 
where a $d$-exponent indicates a DR. 
This query is known to be safe.
We thus expect our
definition of $\rho(q)$ to find that $\rho(q) = \PP{q}$.  
Ignore that $T^d$ 
is deterministic, then  $\preceq$ has two minimal plans:
$q^{\Delta_1} \datarule R^y(x,y),S(x,y),T^d(y)$, and 
$q^{\Delta_2} \datarule R(x),S(x,y),T^{d x}(x,y)$.
Since $\Delta_2$ dissociates only $T^d$, we now know from \autoref{lemma:DetDissociation} that $\PP{q} = \PP{q^{\Delta_2}}$. 
Thus, by using $\preceq$ as before,
we still get the correct answer.
However, evaluating the plan $P^{\Delta_1}$ 
is always unnecessary since
$\Delta_2 \preceq^p \Delta_1$. 
In contrast, without information about DRs, $\Delta_2 \not \preceq^p \Delta_1$, and we would thus have to evaluate both plans. 

\Autoref{Fig_PartialDissociationOrderSimple} illustrates this with augmented incidence matrices: dissociated variables in DRs are now marked with empty circles ($\circ$) instead of full circles ($\bullet$),
and the preorder $\preceq^p$ is determined entirely by full circles (representing dissociated variables in probabilistic relations). 
However, as before, the correspondence to plans (as implied by the hierarchy between all variables) is still determined 
by empty and full circles.
\Autoref{Fig_PartialDissociationOrderSimple_b} shows that 
$\rho(q) = \PP{q^{\Delta_2}} = \PP{q}$
since
$\Delta_0 \equiv^p \Delta_2 \preceq^p \Delta_1 \equiv \Delta_3$.
Thus, the query is safe, and it suffices to evaluate only $P^{\Delta_2}$.
Notice that $q$ is not hierarchical, but still safe since it is in an equivalence class with a query that is hierarchical: $\Delta_0 \equiv^p \Delta_2$.
\Autoref{Fig_PartialDissociationOrderSimple_c} shows that, with $R^d$ and $T^d$ being deterministic, all three possible query plans (corresponding to $\Delta_1$,  $\Delta_2$, and  $\Delta_3$) form a ``\emph{minimal equivalence class}'' in $\preceq^p$ with $\Delta^0$, and thus give the exact probability.
We, therefore, want to modify our algorithm to return just one plan from each ``\emph{minimal safe equivalence class}.'' Ideally, we prefer the plan corresponding to $\Delta_3$ (or more generally, the top plan in $\preceq$ 
for each minimum equivalence class)
since $P^{\Delta_3}$ least constrains the join order between tables.
\end{example}

\begin{figure}[t]
	\vspace{-3mm}
    \centering
	\subfloat[]{
		\label{Fig_PartialDissociationOrderSimple_a}
		\includegraphics[scale=0.43]{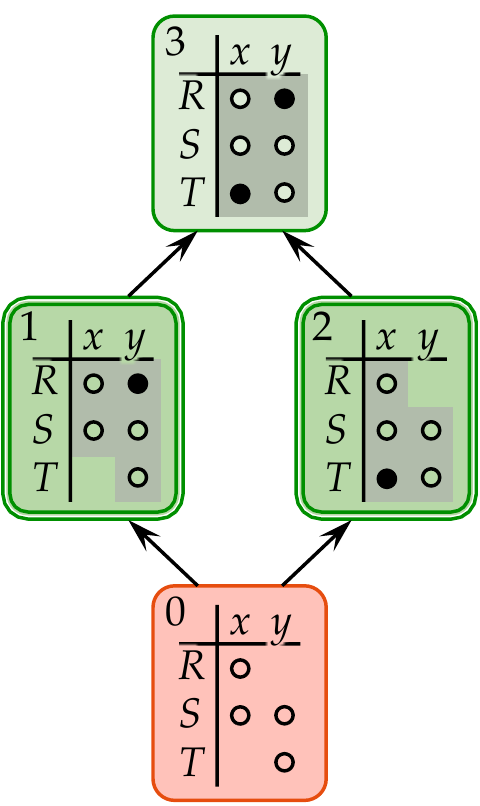}}
	\hspace{3mm}	
	\subfloat[{$T^d$}]{
		\includegraphics[scale=0.43]{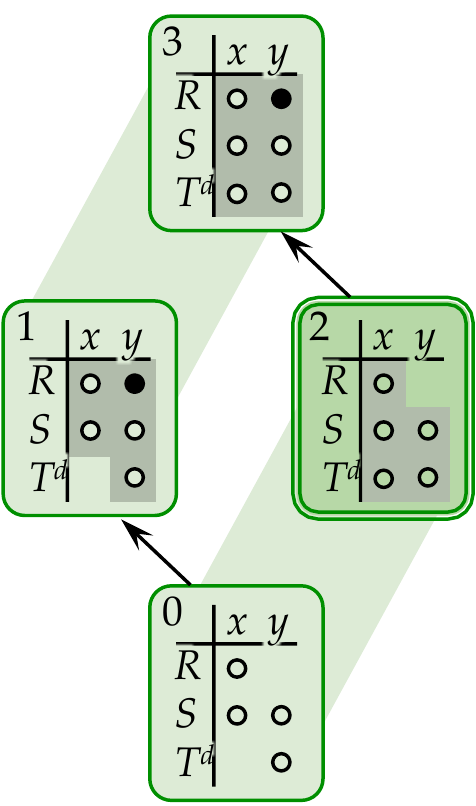}
		\label{Fig_PartialDissociationOrderSimple_b}}
	\hspace{3mm}	
	\subfloat[$R^d$ and $T^d$]{
		\includegraphics[scale=0.43]{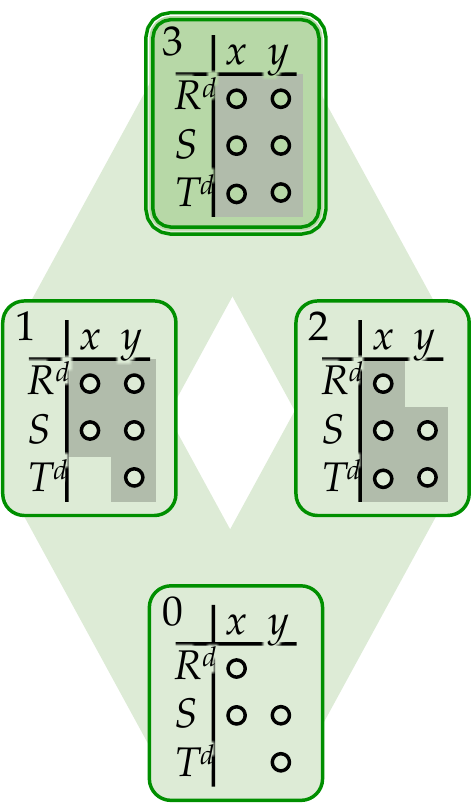}
		\label{Fig_PartialDissociationOrderSimple_c}}		
	\caption{\Autoref{ex:DRsSimple}:
	The presence of DRs $R^d$ and $T^d$ in (b) and (c) changes the original partial dissociation order for $q\datarule R(x),S(x,y),T(y)$ in (a):
	Several dissociations now have the same probability (shown with shaded areas instead of arrows). Our modified algorithm now returns, for \emph{each minimal safe equivalence class}, the query plan for the top most dissociation (shown in dark green and double-lined).
		}
	\label{Fig_PartialDissociationOrderSimple}
	\vspace{-1mm}
\end{figure}

We now explain two simple modifications to \autoref{alg:basicAlgorithm} that achieve exactly our desired optimizations described above:
\begin{enumerate}[nolistsep,label=(\arabic*)]

\item Denote with $\MinPCuts(q)$ the set of minimal cut-sets that disconnect the query into at least two connected components with probabilistic tables. 
Replace $\MinCuts(q)$ in \autoref{alg1:line10} with $\MinPCuts(q)$.

\item 
Denote with $m_p$ the number of probabilistic relations in a query.
Replace the stopping condition in \autoref{alg1:line1} with:
\textbf{if~}{$m^p \leq 1$} \textbf{then~}{$\mathcal{P} \leftarrow \{\projdd{\vec x} \joinp{}{R_1(\vec x_1), \ldots, R_m(\vec x_m)}\}$}. In other words, if a query has maximal one probabilistic relation, than 
join all relations followed by projecting on the head variables.

\end{enumerate}

\begin{theorem}[\autoref{alg:basicAlgorithm} with DRs]\label{prop:algorithm2}
	\Autoref{alg:basicAlgorithm} with above 2 modifications
	returns a minimum number of plans to calculate $\rho(q)$ given schema knowledge about DRs.
\end{theorem}

For example, for $q\datarule R(x),S(x,y),T^d(y)$, $\MinCuts(q) = \{\{x\}, \{y\}\}$, while $\MinPCuts(q) = \{\{x\}\}$. 
Therefore, the modified algorithm returns $P^{\Delta_2}$ as single plan.
For $q\datarule R^d(x),S(x,y),T^d(y)$, the stopping condition is reached (also, $\MinPCuts(q) = \{ \emptyset \}$) and the algorithm returns $P^{\Delta_3}$ as single plan (see \autoref{Fig_PartialDissociationOrderSimple_c}).

\subsubsection{Functional dependencies (FDs)}\label{sec:FDOptimization}

\noindent
Knowledge of functional dependencies (FDs), such as keys, can also restrict
the number of necessary minimal plans.
A well known example is the query $q \datarule R(x), S(x,y), T(y)$ from \autoref{ex:DRsSimple}; it becomes safe 
if we know that $S$ satisfies the FD $\Gamma: x \rightarrow y$ and has a unique safe plan that corresponds to dissociation $\Delta_2$. In other words, we would like our modified algorithm to take $\Gamma$ into account and to not return the plan corresponding to dissociation $\Delta_1$.

Let $\bm{\Upgamma}$ be the set of FDs on $\Var(q)$
consisting of the union of FDs on every atom
$R_i$ in $q$.  As usual, denote ${\vec x}_i^+$ the closure of a set of
attributes ${\vec x}_i$, and denote $\Delta_{\bm{\Upgamma}} = ({\vec y}_1, \ldots,
{\vec y}_m)$ the dissociation defined as follows: for every atom
$R_i({\vec x}_i)$ in $q$, ${\vec y}_i = {\vec x}_i^+ \setminus
{\vec x}_i$. Then we show:

\begin{lemma}[Dissociation and FDs]\label{lemma:FDs1}
Dissociating a table $R_i$ on any variable $y \in {\vec x}_i^+$ does not change the probability.
\end{lemma}

\noindent
This lemma is similar to \autoref{lemma:DetDissociation}. We can thus further refine our probabilistic dissociation preorder ${{\preceq^p}'}$ by:
\begin{align*}
  	\Delta {\preceq^p}' \Delta \Leftrightarrow \forall i, R_i \mbox{ probabilistic}: 
	{\vec y}_i \setminus \vec x_i^+ \subseteq {\vec y}_i' \setminus \vec x^+_i
\end{align*}

\noindent
As a consequence, using ${\preceq^p}'$ instead of $\preceq^p$, allows us to further reduce the number of minimal safe equivalence classes. 
We next state a result by \cite{DBLP:conf/icde/OlteanuHK09} in our notation:

\begin{proposition}[Safety and FDs~\protect{\cite[Prop.~IV.5]{DBLP:conf/icde/OlteanuHK09}}] A query $q$
  is safe iff $q^{\Delta_{\bm{\Upgamma}}}$ is hierarchical.
\end{proposition}

This justifies our third modification to \autoref{alg:basicAlgorithm} for computing 
$\rho(q)$ of a query $q$ over a database that
satisfies $\bm{\Upgamma}$: First compute ${\Delta_{\bm{\Upgamma}}}$, then run $q^{\Delta_{\bm{\Upgamma}}}$ on our previously modified \autoref{alg:basicAlgorithm}.

\begin{theorem}[\autoref{alg:basicAlgorithm} with FDs]\label{prop:allQueryPlansFDs}
	\Autoref{alg:basicAlgorithm} with above 3 modifications
	returns a minimum number of plans to calculate $\rho(q)$ given schema knowledge about DRs and FDs.
\end{theorem}

It is easy to see that our modified algorithm returns one single
plan iff the query is safe, taking into account its structure,
DRs and FDs. It is thus a
\emph{strict generalization of all known safe self-join-free
  conjunctive
  queries}~\cite{DBLP:journals/vldb/DalviS07,DBLP:conf/icde/OlteanuHK09}.
In particular, we can reformulate the known safe query dichotomy~\cite{DBLP:journals/vldb/DalviS07} in our notation very succinctly:

\begin{corollary}[Dichotomy]
	$\PP{q}$ can be calculated in PTIME iff there exists a dissociation $\Delta$ of $q$ that is ($i$) hierarchical, and ($ii$) in an equivalence class with $q$ under 
	${\preceq^p}'$. 
\end{corollary}

\noindent
To see what the corollary says, assume first that there are no FDs: Then $q$ is in PTIME iff there exists a dissociation $\Delta$ of the DRs only, such that $q^\Delta$ is hierarchical.  If there are FDs, then we first compute the full dissociation $\Delta_{\bm{\Upgamma}}$ (called ``full chase'' in \cite{DBLP:conf/icde/OlteanuHK09}), then apply the same criterion to $q^{\Delta_{\bm{\Upgamma}}}$.

%% file: 4optimizingDissociation.tex
\section{Multi-query Optimizations}\label{sec:optimizations}
\noindent
So far, \autoref{alg:basicAlgorithm} enumerates all minimal query plans. 
We then take the minimum score of those plans in order to calculate the propagation score $\rho(q)$.
In this section, we develop three optimizations that can considerably reduce 
the necessary calculations for evaluating all minimal query plans. Note that these three optimizations and the two optimizations from the previous section are orthogonal and can be arbitrarily combined in the obvious way.
We use the following example to illustrate the first two optimizations.

\begin{example}[Optimizations]\label{ex:optimizationExample}
Consider 
$q \datarule$ $R(x,z),S(y,u),$ $T(z), U(u), M(x,y,z,u)$.
Our default is to evaluate all 6 minimal plans returned by \specificref{Algorithm}{alg:basicAlgorithm}, 
then take the minimum score (shown in \figref{Fig_AllOptimizedQueryPlans}{a}).
\Figref{Fig_AllOptimizedQueryPlans}{b}
and \figref{Fig_AllOptimizedQueryPlans}{c}
illustrate the optimized evaluations after applying Opt.~1, or Opt.~1 and Opt.~2, respectively.
\markend
\end{example}

\subsection{Opt.\ 1: One single query plan}\label{sec:oneSinglePlan}
\noindent
Our first optimization creates one single query plan by \emph{pushing the min-operator down into the leaves}. It thus avoids calculations when it is clear that other calculations must have lower bounds. The idea is simple: Instead of creating one query subplan for each top set  $\vec y \in \TopSet(q)$ in \autoref{alg1:line12} of \autoref{alg:basicAlgorithm}, the adapted \autoref{alg:opt1Algorithm} takes the minimum score over those top sets, for each tuple of the head variables in \autoref{alg2:line11}. It thus creates one single query plan.

\begin{figure}[t]
\centering
{\includegraphics[scale=0.82]{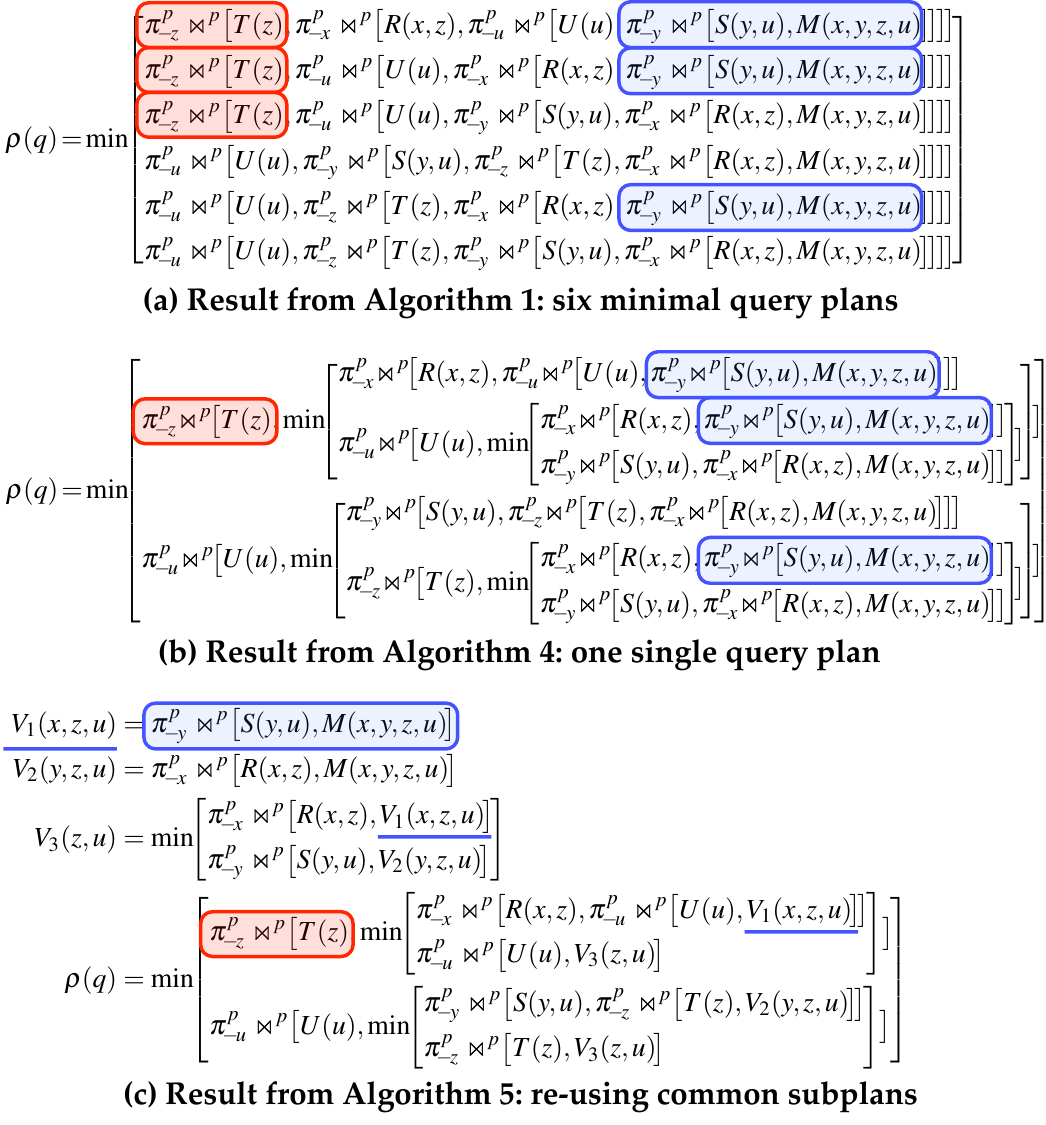}}
\vspace{-4mm}
\caption{
\Autoref{ex:optimizationExample} before and after applying optimizations 1 and 2.
  }
\label{Fig_AllOptimizedQueryPlans}
\end{figure}

\begin{algorithm2e}[t]
\scriptsize
\caption{Optimization~1 recursively pushes the min operator into the leaves and generates one single query plan.}\label{alg:opt1Algorithm}
\SetKwInput{Algorithm}{Recursive algorithm}
\SetKwFunction{SP}{SP}
\SetKwFor{ForAll}{forall}{do}{endfch}	
\LinesNumbered

\Algorithm{\FuncSty{SP (SinglePlan)}}
\KwIn{Query $q(\vec x) \datarule R_1(\vec x_1), \ldots, R_m(\vec x_m)$}	
\KwOut{Single query plan $P$}
\BlankLine
		
\lIf{$m = 1$}{$P \leftarrow  \projpd{\vec x} R_i(\vec x_i)$}
\Else{
	\uIf{$q$ is disconnected}{
		Let $q = q_1, \ldots, q_k$ be the components connected by $\EVar(q)$ \;
		Let $\HVar(q_i) \leftarrow \HVar(q) \cap \Var(q_i)$ \;
		$P \leftarrow \joinp{}{\SP (q_1 ), \ldots, \SP (q_k )}$ \;
	}
	\Else{
		Let $\TopSet(q) = \{\vec y_1, \ldots, \vec y_j\}$ \;
		Let $q'_i \leftarrow q_i$ with $\HVar(q'_i) \leftarrow \HVar(q) \cup \vec y_i$ \;
		\textbf{if} $j \!=\! 1$ \textbf{then} 
			$P \leftarrow \projp{\vec y_1} \SP (q_1')$ \;
		\textbf{else} 
			$P \leftarrow  \minp{\projp{\vec y_1} \SP (q_1'), \ldots, 
			\projp{\vec y_j} \SP (q_j')}$ \;\label{alg2:line11}
	}
}
\end{algorithm2e}

\subsection{Opt.~2: Re-using common subplans}\label{sec:reusingCommonSubplans}
\noindent
Our second optimization calculates only once, then \emph{re-uses common subplans shared between the minimal plans}.
Thus, whereas our first optimization reduces computation by combining plans at their roots, the second optimization stores and re-uses common results in the branches.
The adapted \autoref{alg:opt2Algorithm} works as follows: It first traverses the whole single query plan (\FuncSty{FindingCommonSubplans}) and remembers each subplan by the atoms used and its head variables in a HashSet \FuncSty{HS} (\autoref{alg3:line14}). If it sees a subplan twice (\autoref{alg3:line12}), it creates a new view for this subplan, mapping the subplan to a new view definition. The actual plan (\FuncSty{ViewReusingPlan}) then uses these views whenever possible (\autoref{alg3:line17}). The order in which the views are created (\autoref{alg3:line5}) assures that the algorithm also discovers and exploits \emph{nested common subexpressions}. \Figref{Fig_AllOptimizedQueryPlans}{c} illustrates for \autoref{ex:optimizationExample}, that both the main plan and the view $V_3$ re-use views $V_1$ and $V_2$.

\begin{algorithm2e}[t]
\scriptsize
\caption{Optimizations 1 \& 2 together create a query plan which re-uses several previously defined temporary views.}\label{alg:opt2Algorithm}
\SetKwInput{Algorithm}{Algorithm}
\SetKwInput{Function}{Recursive function}
\SetKwFunction{FS}{FS}
\SetKwFunction{SP}{SP}
\SetKwFunction{RP}{RP}
\SetKwFunction{HS}{HS}
\SetKwFunction{HM}{HM}
\SetKwFor{ForAll}{forall}{do}{endfch}	
\LinesNumbered

\Algorithm{\FuncSty{UsingCommonSubplans}}
\KwIn{Query $q(\vec x) \datarule R_1(\vec x_1), \ldots, R_m(\vec x_m)$}	
\KwOut{Ordered set of view definitions $\mathcal V$, final query plan $P$}
$\HS\leftarrow \emptyset$  \qquad // HashSet of all subplans\;
$\HM\leftarrow \!(\emptyset,\emptyset)$  \, // HashMap from subplans to unique view names \;
$\mathcal V  \leftarrow \emptyset$  \qquad\,\!\! // Set of view definitions \;
$\FS(q)$ \;
\ForEach{$q_i \in \HM$.keys in increasing size of $\HVar(q_i)$ and $\Var(q_i)$\label{alg3:line5}}{
	$\mathcal V \leftarrow \mathcal V \cup \{\HM.val = \FuncSty{ViewReusingPlan}(q_i) \}$ \;
}
$P = \RP(q)$ \;
\BlankLine

\BlankLine
\Function{\FuncSty{FS (FindingCommonSubplans)}}
\KwIn{Query $q(\vec x) \datarule R_1(\vec x_1), \ldots, R_m(\vec x_m)$}	
\uIf{$q$ is disconnected}{
	Let $q = q_1, \ldots, q_k$ be the components connected by $\EVar(q)$ \;
	\lForEach{$q_i$}{$\FS(q_i(\vec x_i))$}
}
\Else{
\lIf{$(m = 1 \wedge \vec x = \vec x_i) \vee \HM(q) \neq \emptyset$}{
	\Return
}
\lIf{$q \in \HS \wedge \HM(q) = \emptyset $\label{alg3:line12}}{
	$\HM(q) \leftarrow \textrm{new view name}$ 
}
	$\FuncSty{HS} \leftarrow \FuncSty{HS} \cup \{q\}$ \;\label{alg3:line14}
	\ForEach{$\vec y \in \TopSet(q)$}{
		Let $q' \leftarrow q$ with $\HVar(q') \leftarrow \HVar(q) \cup \vec y$ \;
		$\FS(q')$ \;
}
}
\BlankLine
\BlankLine
\Function{\FuncSty{RP (ViewReusingPlan)}}
\KwIn{Query $q(\vec x) \datarule R_1(\vec x_1), \ldots, R_m(\vec x_m)$}	
\KwOut{Query plan $P$ that reuses views from HashMap \FuncSty{HM}}
\textbf{if} $\HM(q) \neq \emptyset $ \textbf{then} $P \leftarrow  \HM(q)$ \;\label{alg3:line17}
\Else{
	\textit{Insert here lines 1-11 from \autoref{alg:opt1Algorithm}, replacing \SP with \RP} \;
}
\end{algorithm2e}

\subsection{Opt.~3: Deterministic semi-join reduction}\label{sec:opt3}
\noindent
The most expensive operations in probabilistic query plans are the group-bys for the probabilistic project operations. These are often applied early in the plans to tuples which are later pruned and do not contribute to the final query result. 
Our third optimization is to first apply a \emph{full semi-join reduction on the input relations} before starting the probabilistic evaluation from these \emph{reduced input relations}. 
We like to draw here an important connection to \cite{DBLP:conf/icde/OlteanuHK09}, which introduces the idea of ``lazy plans'' and shows orders of magnitude performance improvements for safe plans by computing confidences not after each join and projection, but rather at the very end of the plan.
We note that our semi-join reduction \emph{serves the same purpose} with similar performance improvements and also apply for safe queries.
The advantage of semi-join reductions, however, is that we do not require any modifications to the query engine.

%% file: 5experiments.tex
\section{Experiments}\label{sec:experiments}
\noindent
We are interested in both the efficiency (``how fast?'') and the quality (``how good?'') of ranking by dissociation as compared to 
exact probabilistic inference, 
Monte Carlo simulation (MC), 
and standard deterministic query evaluation (``deterministic SQL'').

\introparagraph{Ranking quality}
We use \textit{mean average precision} (MAP) to evaluate the quality of a ranking by comparing it against the ranking from exact probabilistic inference as ground truth (GT). 
MAP rewards rankings that place relevant items earlier;
the best possible value is 1, and the worst possible 0.
We use a variant of ``Average Precision at 10'' defined as $\mathrm{AP}@10 := \frac{\sum_{k=1}^{10} \mathit{\mathrm{P}@k}}{10}$. Here, P@$k$ is the precision at the $k$th answer, i.e., the fraction of top $k$ answers according to GT that are also in the top $k$ answers returned.
Averaging over several experiments yields MAP~\cite{Manning:2008:IIR:1394399}.
We use a variant of the analytic method proposed in~\cite{McSherryN08} to  calculate $\mathrm{AP}$ in the presence of ties.
As baseline for no ranking, we assume all tuples have the same score and are thus tied for the same position. We call this baseline ``\emph{random average precision}.''

\introparagraph{Exact probabilistic inference} 
Whenever possible, we calculate GT rankings with a tool called  
SampleSearch~\cite{DBLP:conf/uai/GogateD10,SampleSearch},
which also serves to evaluate the cost of \emph{exact} probabilistic inference.
We describe the method of transforming the lineage DNF into a format that can be read by SampleSearch  in~\cite{DBLP:journals/tods/GatterbauerS14}.

\introparagraph{Monte Carlo (MC)}
We evaluate the MC simulations for different numbers of samples and write MC($x$) for $x$ samples. For example, AP for MC(10k) is the result of sampling the individual tuple scores 10\,000 times from their lineages and then evaluating AP once over the sampled scores. The MAP scores together with the standard deviations are then the average over several repetitions.

\introparagraph{Ranking by lineage size}
To evaluate the potential 
of non-probabi\-listic methods for ranking answers,
we also rank the answer tuples by decreasing size of their lineages; i.e.\ number of terms.
Intuitively, a larger lineage size should indicate that an answer tuple has more ``support'' and should thus be more important.

\introparagraph{Setup 1}
We use the {TPC-H} DBGEN data generator~\cite{tpc-h}
to generate a 1GB database to which we add a column \sql{P} for each table and store it in PostgreSQL 9.2~\cite{postgresql}.
We assign to each input tuple $i$ 
a random probability $p_i$ uniformly chosen from the interval $[0, p_{i\max}]$, resulting in an expected average input probability $\avg[p_{i}] = p_{i\max}/2$. 
By using databases with $\avg[p_{i}] < 0.5$, we can avoid output probabilities close to $1$ for queries with very large lineages. 
We use the following parameterized query:
\begin{align*}
	&Q(a)  \datarule S(\underline{s},a), \mathit{PS}(s,u), P(\underline{u},n), s \leq \$1, n \sql{ like } \$2 \\
&
\parbox{30mm}{
\vspace{-1mm}
	\sql{
		\begin{tabbing}
		\hspace{3mm}\=\hspace{3mm}\=\hspace{3mm}\=\hspace{0cm}\=\hspace{0cm}\=\kill
		select distinct s\_nationkey	
		from Supplier, Partsupp, Part	\\
		where s\_suppkey = ps\_suppkey	
		and ps\_partkey = p\_partkey	\\
		and s\_suppkey $<=$ \$1			
		and p\_name like \$2	
		\end{tabbing}
\vspace{-3mm}
}}
\notag
\end{align*}
Parameters {\$1} and {\$2} allow us to 
change the lineage size.
Tables \sql{Supplier}, \sql{Partsupp} and \sql{Part} have 10k, 800k and 200k tuples, respectively.
There are 25 different numeric attributes for \sql{nationkey} and our goal is to efficiently rank these 25 nations. 
As baseline for not ranking, we use random average precision for 25 answers, which leads to MAP@10 $\approx 0.220$.
This query has two minimal query plans and we will compare the speed-up from either evaluating both individually or performing a deterministic semi-join reduction (Optimization 3) on the input tables.

\introparagraph{Setup 2}
We compare the run times for our three optimizations against evaluation of all plans for $k$-chain queries and $k$-star queries over varying database sizes (to evaluate data complexity) and varying query sizes (to evaluate query complexity): 
\begin{align*}
	&{\textrm{$k$-chain: }}q(x_0, x_k)  \datarule R_1(x_0, x_1), R_2(x_1,x_2), \ldots, R_k(x_{k-1},x_k) \\
	&{\textrm{$k$-star: }}q('\!a')	 \datarule R_1('\!a', x_1), R_2(x_2), \ldots, R_k(x_k), R_0(x_1, \ldots, x_k)
\end{align*}

\noindent
We denote the length of the query with $k$, the number of tuples per table with $n$, and the domain size with $N$. We use integer values which we uniformly draw from the range $\{1,2, \ldots N\}$. 
Thus, the parameter $N$ determines the \emph{selectivity} and is varied as to 
keep the answer cardinality constant around 20-50 for chain queries, or the
answer probability between 0.90 and 0.95 for star queries.
For the data complexity experiments, we vary the number of tuples $n$ per table between $100$ and $10^6$. 
For the query complexity experiments, we vary $k$ between $2$ and $8$ for chain queries.
For these experiments, the optimized (and often \emph{extremely long}) SQL statements are ``calculated'' in JAVA and then sent to Microsoft SQL server 2012.
To illustrate with numbers, we have to issue 429 query plans in order to evaluate the $8$-chain query (see \autoref{table:numberMinimalQueryPlans}). Each of these plans joins 8 tables in a different order. Optimization 1 then merges those plans together into one truly gigantic single query plan.

\subsection{Run time experiments}\label{sec:exOptimizations}

\begin{question}\label{question11}
When and how much do our three query optimizations speed up query evaluation?
\end{question}

\begin{result}\label{lesson11}
Combining plans (Opt.~1) and using intermediate views (Opt.~2) almost always speeds up query times.
The semi-join reduction (Opt.~3) 
slows down queries with high selectivities, but 
considerably speeds up queries with small selectivities.
\end{result}

\noindent
\Figuresref{Fig_SyntheticChain4}{Fig_SyntheticChainQuerySize}
show the results on setup 2 for increasing database sizes or query sizes.
For example, 
\autoref{Fig_SyntheticChain7} shows the performance of computing a 
7-chain query
which has 132 safe dissociations.  
Evaluating each of these queries separately takes a long time,
while our optimization techniques 
bring evaluation time close to deterministic query evaluation.  
Especially on larger databases, where the running time is I/O bound, the
penalty of the probabilistic inference is only a factor of 2-3 in this example.
Notice here the trade-off between optimization~1,2 and optimization 1,2,3: Optimization 3 applies a full semi-join reduction on the input relations before starting the probabilistic plan evaluation from these reduced input relations. This operation imposes a rather large constant overhead, both at the query optimizer and at query execution. For larger databases (but constant selectivity), this overhead is amortized. 
In practice, this suggests that dissociation allows us a large space of optimizations depending on the query and particular database instance that can conservatively extend the space of optimizations performed today in deterministic query optimizers.

\Figuresref{Fig_VLDBJ_TPCH_timing1}{Fig_VLDBJ_TPCH_timing3}
compare the running times on setup 1 between
dissociation with two minimal query plans  (``Diss''), 
dissociation with semi-join reduction (``Diss + Opt3''), exact probabilistic inference (``SampleSearch''), 
Monte Carlo with 1000 samples (``MC(1k)''),
retrieving the lineage only (``Lineage query''), 
and deterministic query evaluation without ranking (``Standard SQL'').
We fixed $\$2 \in \{ \sql{'\%red\%green\%'}, \sql{'\%red\%'}, \sql{'\%'} \}$ and varied $\$1 \in \{500, 1000, \ldots 10k\}$.
\Autoref{Fig_VLDBJ_TPCH_timing4} combines all three previous plots and shows the times as function of the maximum lineage size (i.e. the size of the lineage for the tuple with the maximum lineage) of a query.
We see here again that the semi-join reduction speeds up evaluation considerably for small lineage sizes
(\autoref{Fig_VLDBJ_TPCH_timing1} shows speedups of up to 36).
For large lineages, however, the semi-join reduction is an unnecessary overhead, as most tuples are participating in the join anyway (\autoref{Fig_VLDBJ_TPCH_timing2} shows overhead of up to 2).

\begin{figure*}[t]
    \centering
	\vspace{5mm}
	\subfloat[$4$-chain queries]
       	{\includegraphics[scale=0.43]{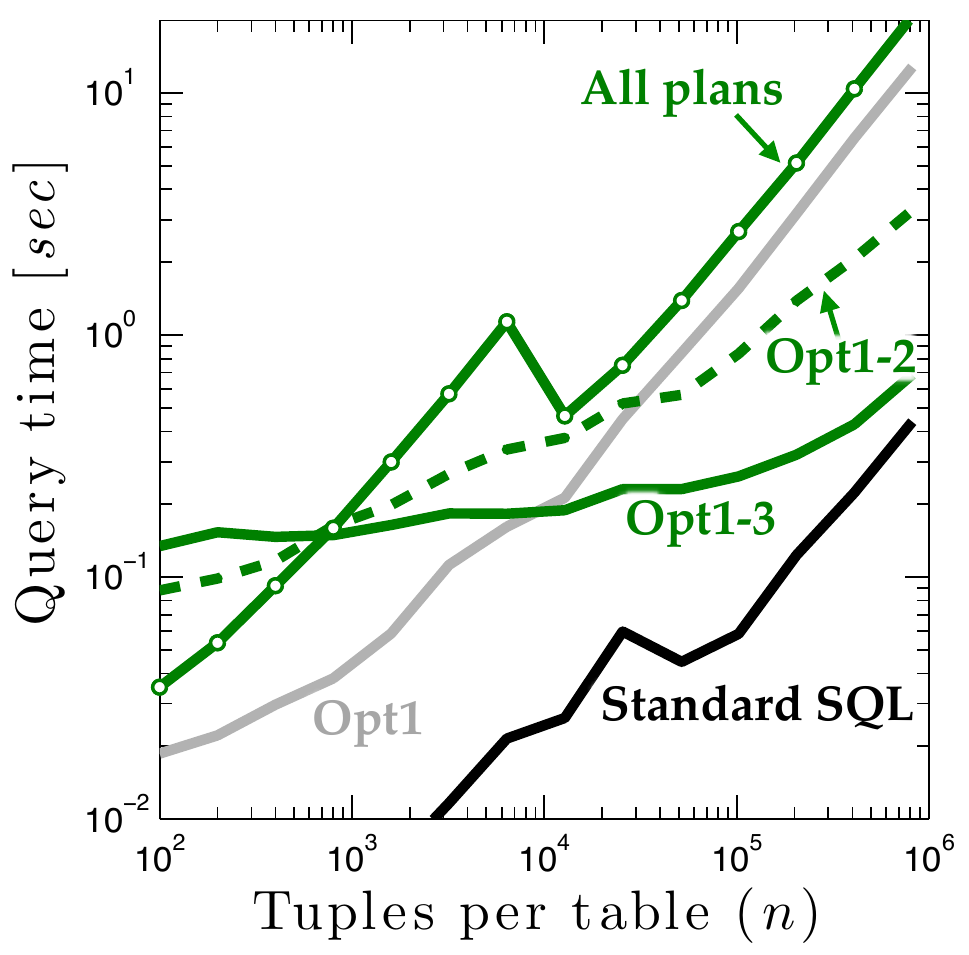}
    	\label{Fig_SyntheticChain4}}
	\hspace{-0.1mm}
	\subfloat[$7$-chain queries]
		{\includegraphics[scale=0.43]{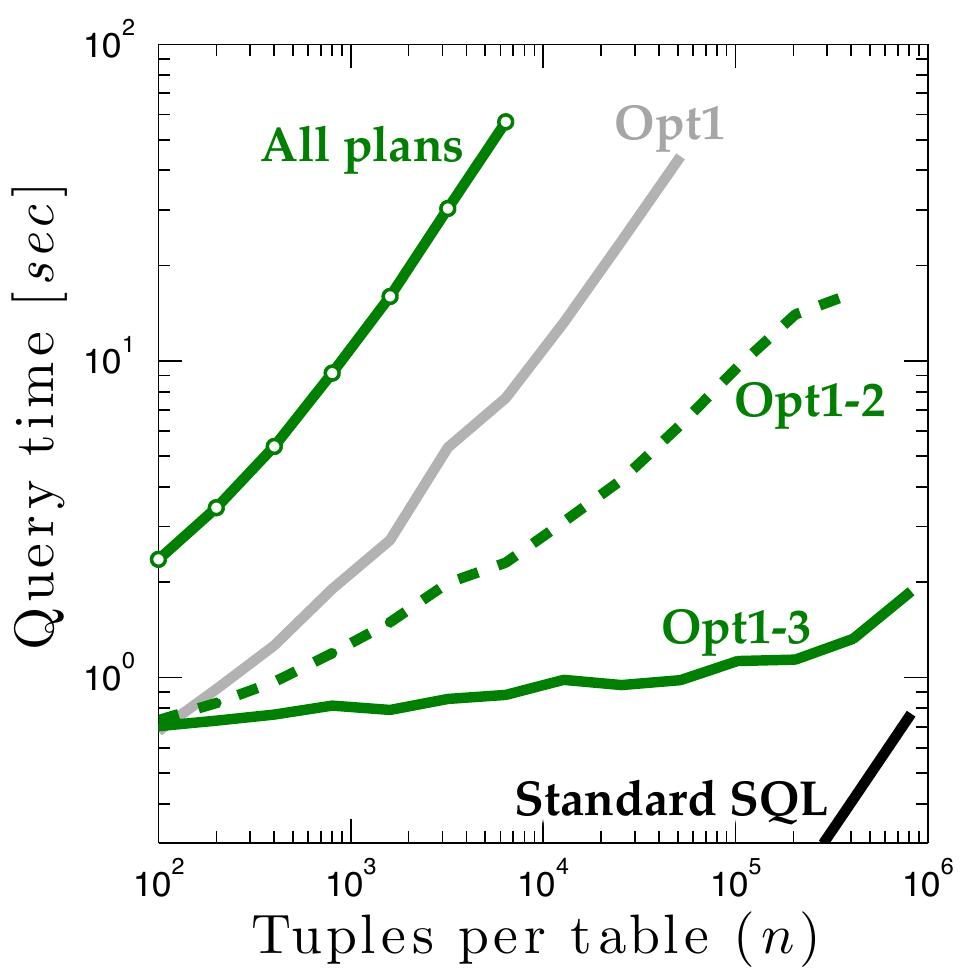}
		\label{Fig_SyntheticChain7}}
	\hspace{-0.1mm}
	\subfloat[$2$-star queries]
		{\includegraphics[scale=0.43]{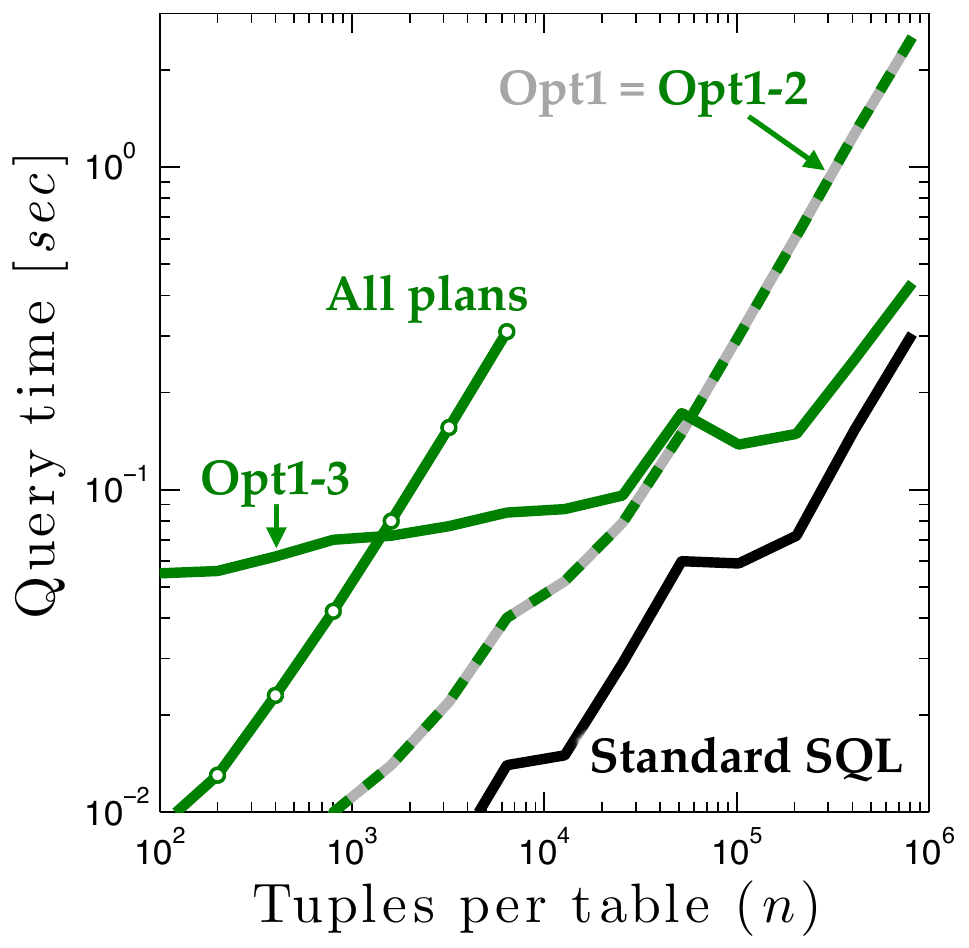}
		\label{Fig_SyntheticStar2}}
    \hspace{-0.1mm}
	\subfloat[$k$-chain queries]
       	{\includegraphics[scale=0.445]{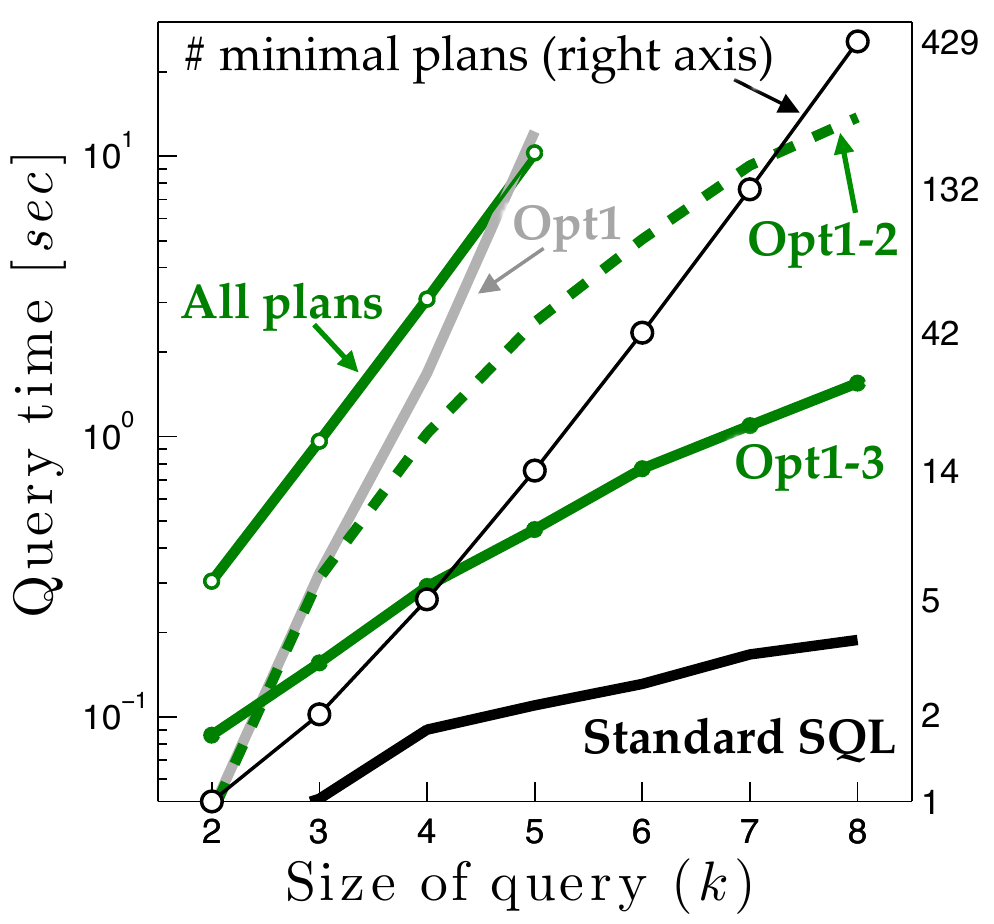}
    	\label{Fig_SyntheticChainQuerySize}}
	\hspace{1mm}
    \subfloat[\$2 = \%red\%green\%]
		{\includegraphics[scale=0.325]{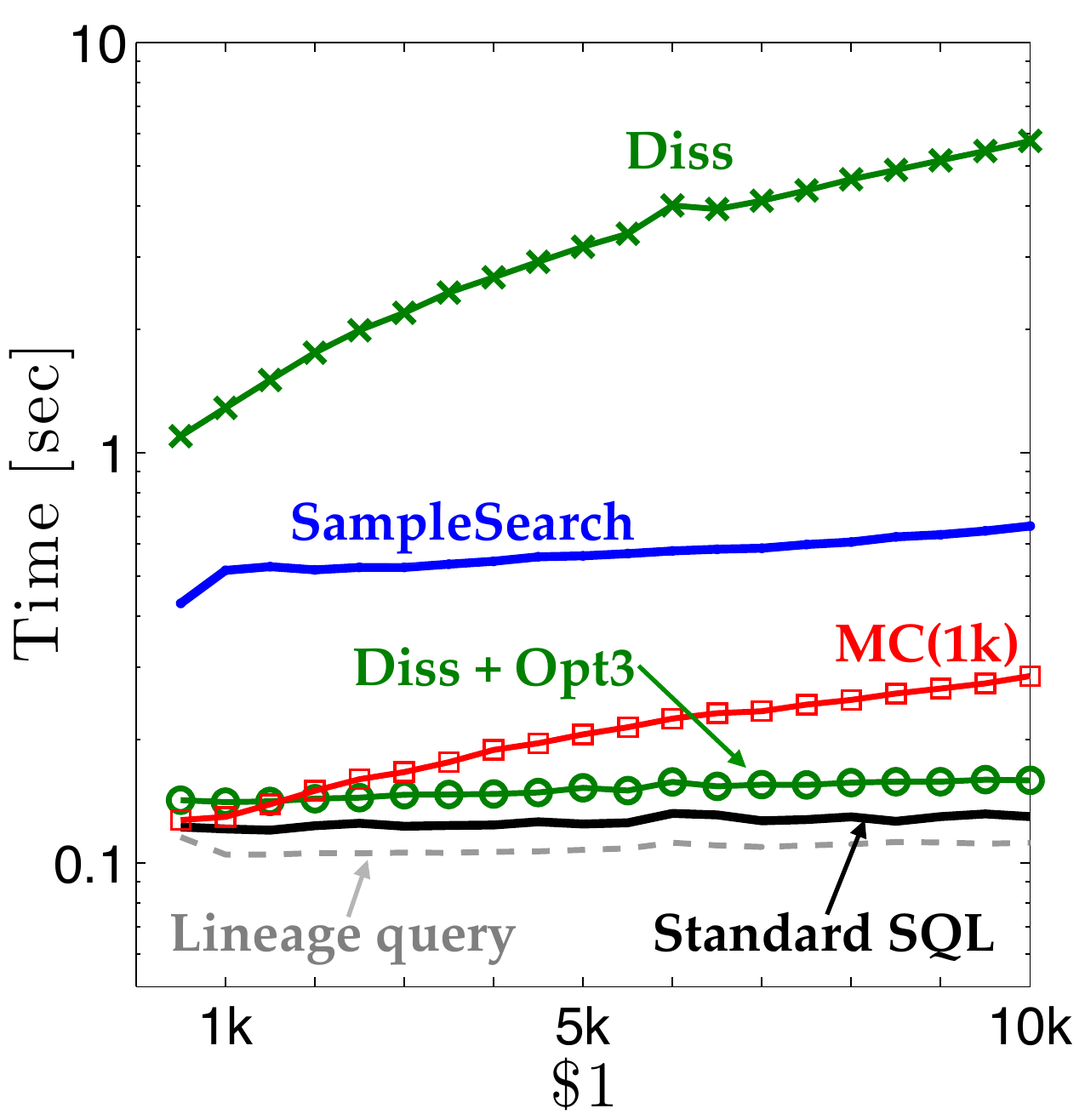}
		\label{Fig_VLDBJ_TPCH_timing1}}					
	\hspace{-0.2mm}
	\subfloat[\$2 = \%red\%]
		{\includegraphics[scale=0.325]{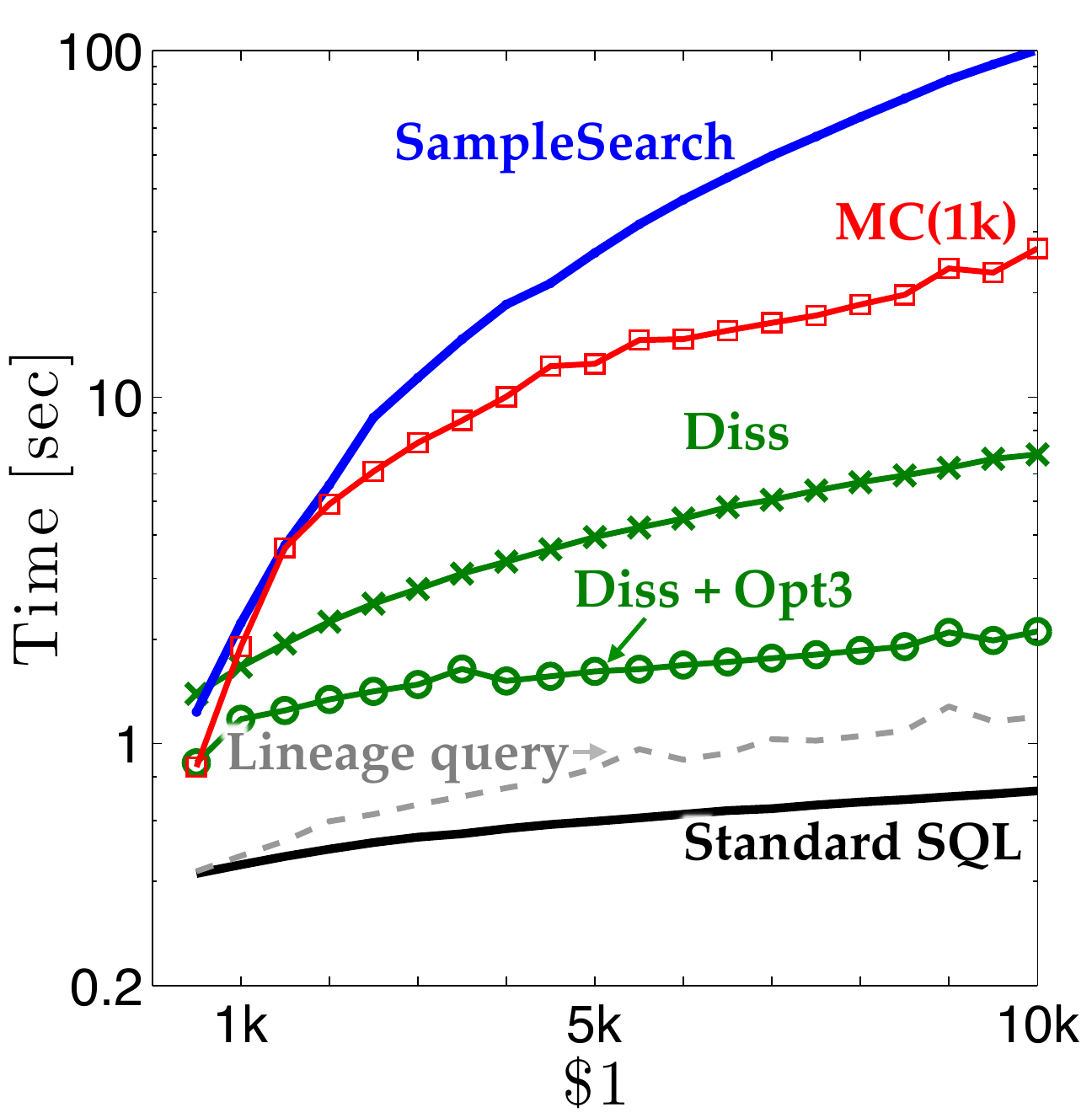}
		\label{Fig_VLDBJ_TPCH_timing2}}
	\hspace{-0.2mm}
    \subfloat[\$2 = \%]
		{\includegraphics[scale=0.325]{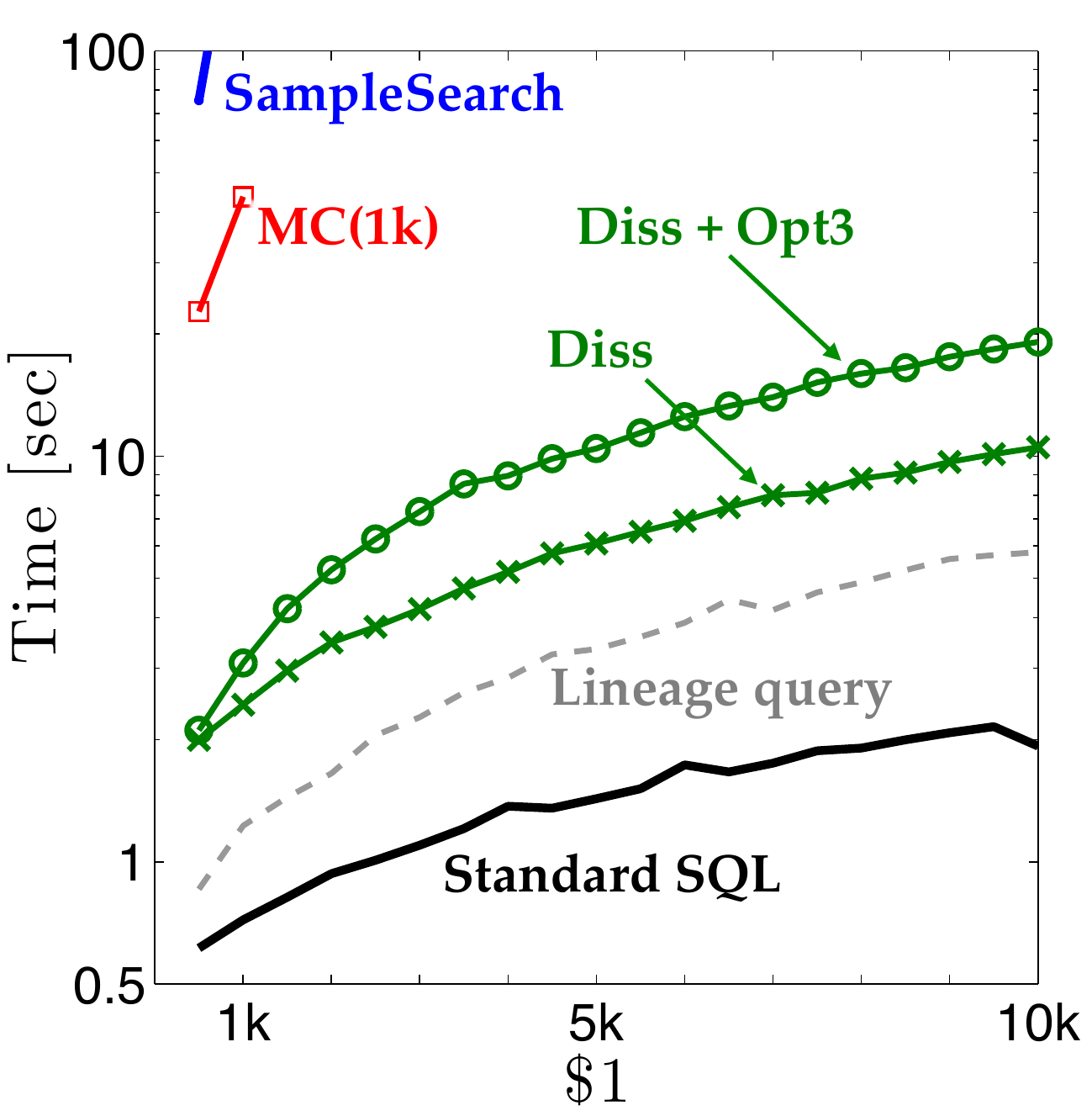}
		\label{Fig_VLDBJ_TPCH_timing3}}	
	\hspace{-0.2mm}
    \subfloat[Combining (a)-(c)]
		{\includegraphics[scale=0.325]{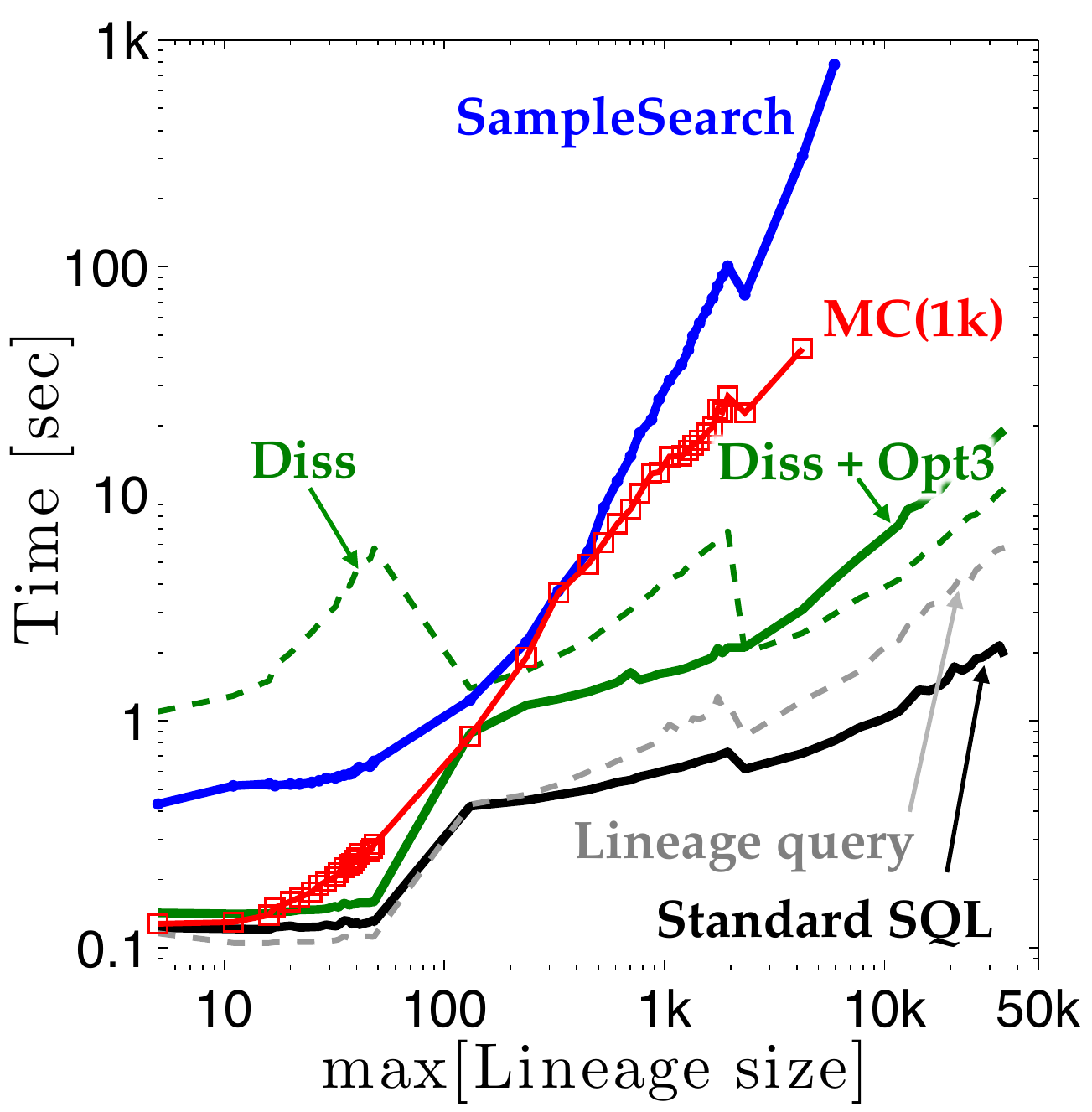}
		\label{Fig_VLDBJ_TPCH_timing4}}		
	\hspace{0.5mm}
	\subfloat[\autoref{lesson1}]
		{\includegraphics[scale=0.325]{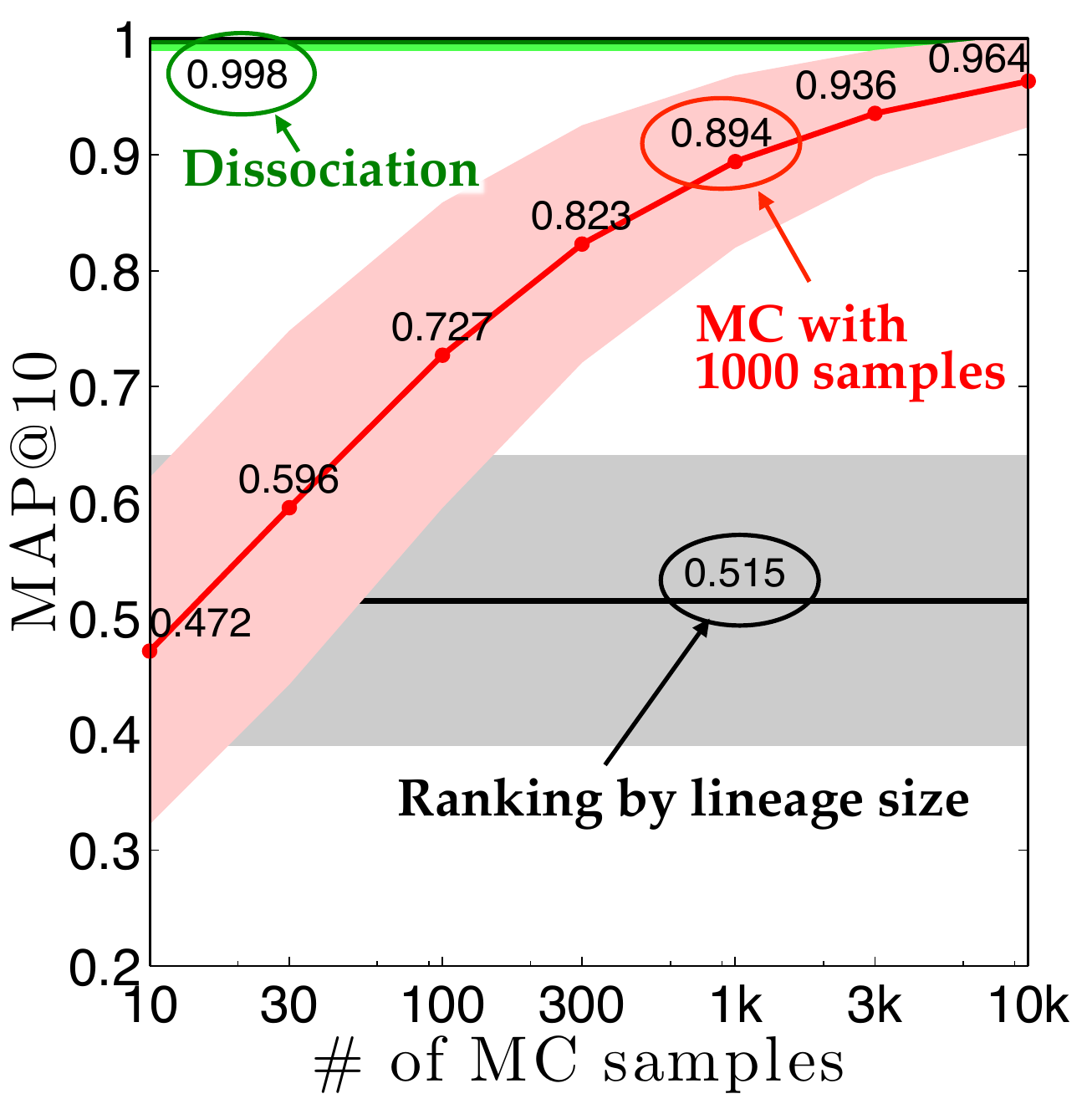}
		\label{Fig_VLDBJ_TPCH_AP_MC_aggregated_09}}
	\hspace{-1mm}
    \subfloat[\autoref{lesson2}]
		{\includegraphics[scale=0.325]{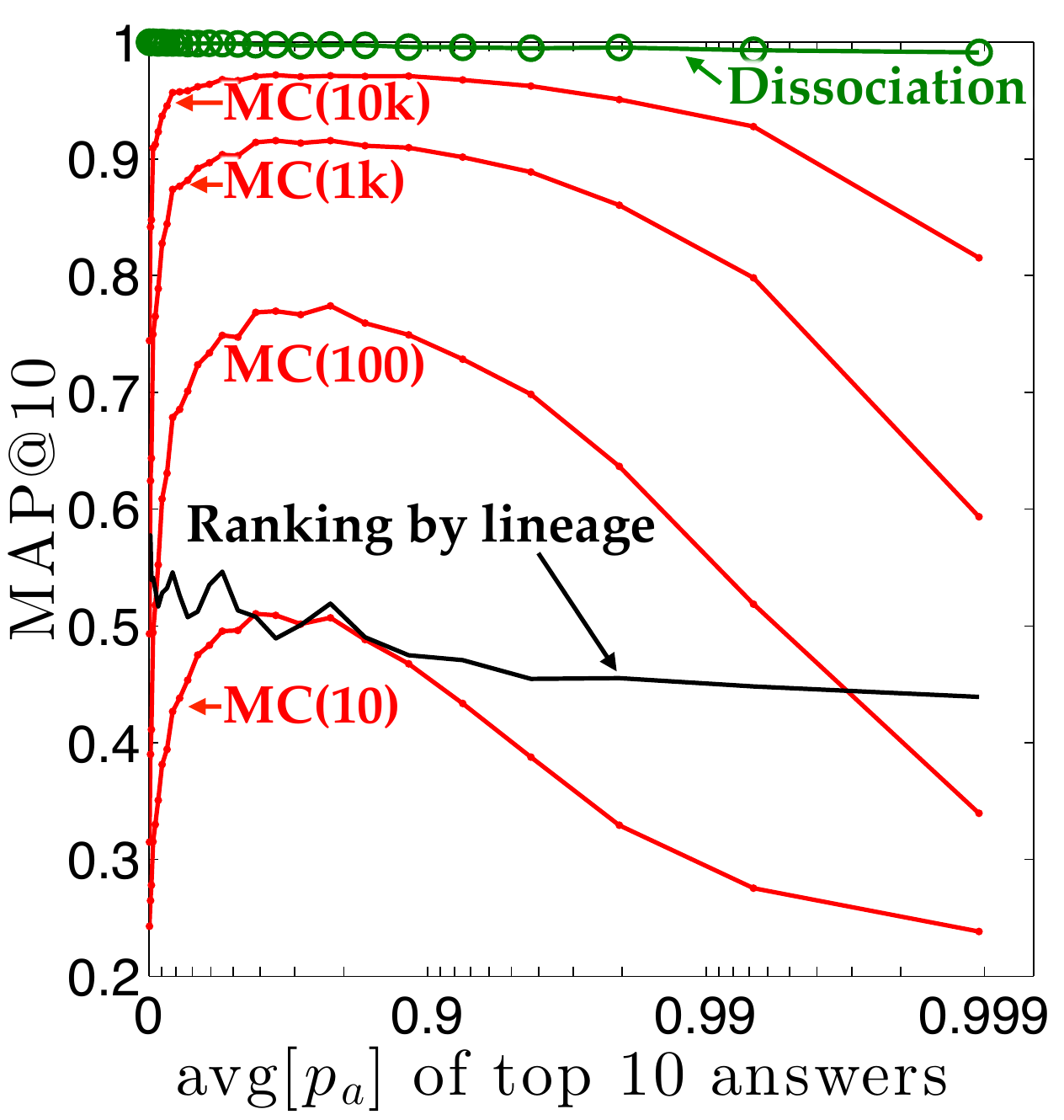}
		\label{Fig_VLDBJ_TPCH_ap_mc_random_redgreen_equalSizedBins_log}}	
	\hspace{-1mm}
    \subfloat[\autoref{lesson3}]
		{\includegraphics[scale=0.325]{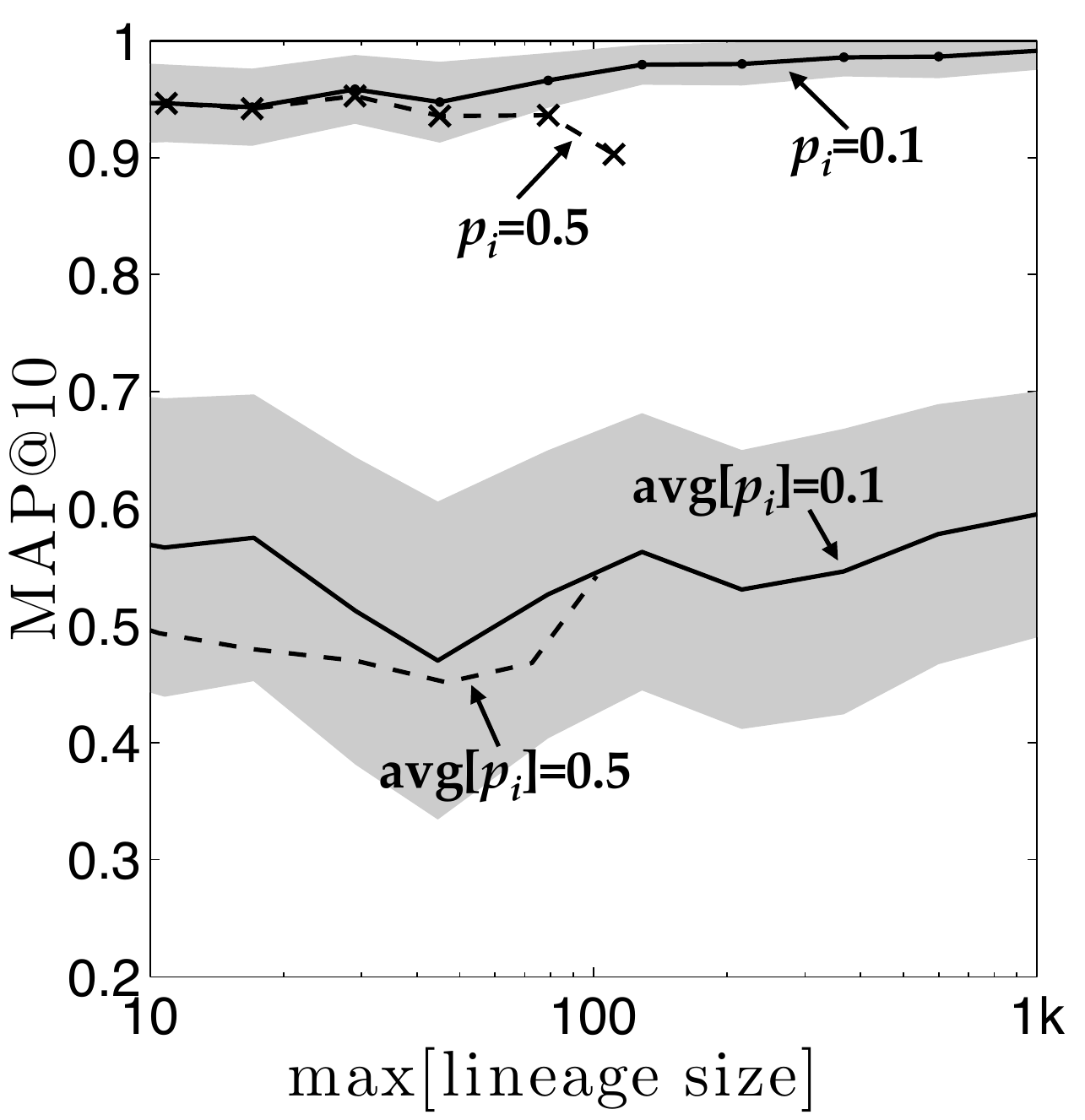}
		\label{Fig_VLDBJ_TPCH_AP_LineageRanking}}
	\hspace{-1mm}
	\subfloat[\autoref{lesson4}]
		{\includegraphics[scale=0.325]{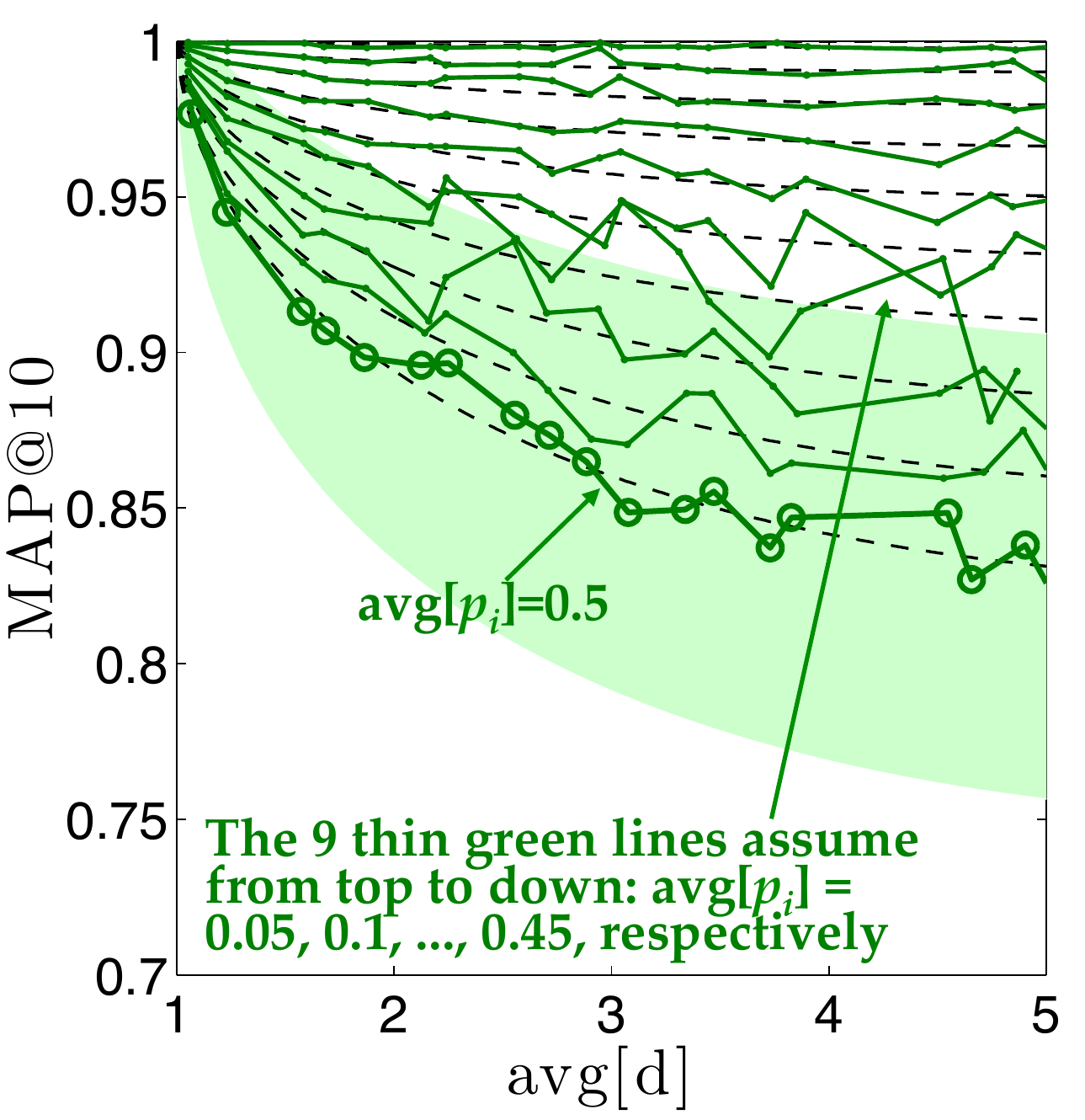}
		\label{Fig_VLDBJ_TPCH_AP_Diss_MC_Tradeoff1}}
	\hspace{-1mm}
	\subfloat[\autoref{lesson4}]
		{\includegraphics[scale=0.325]{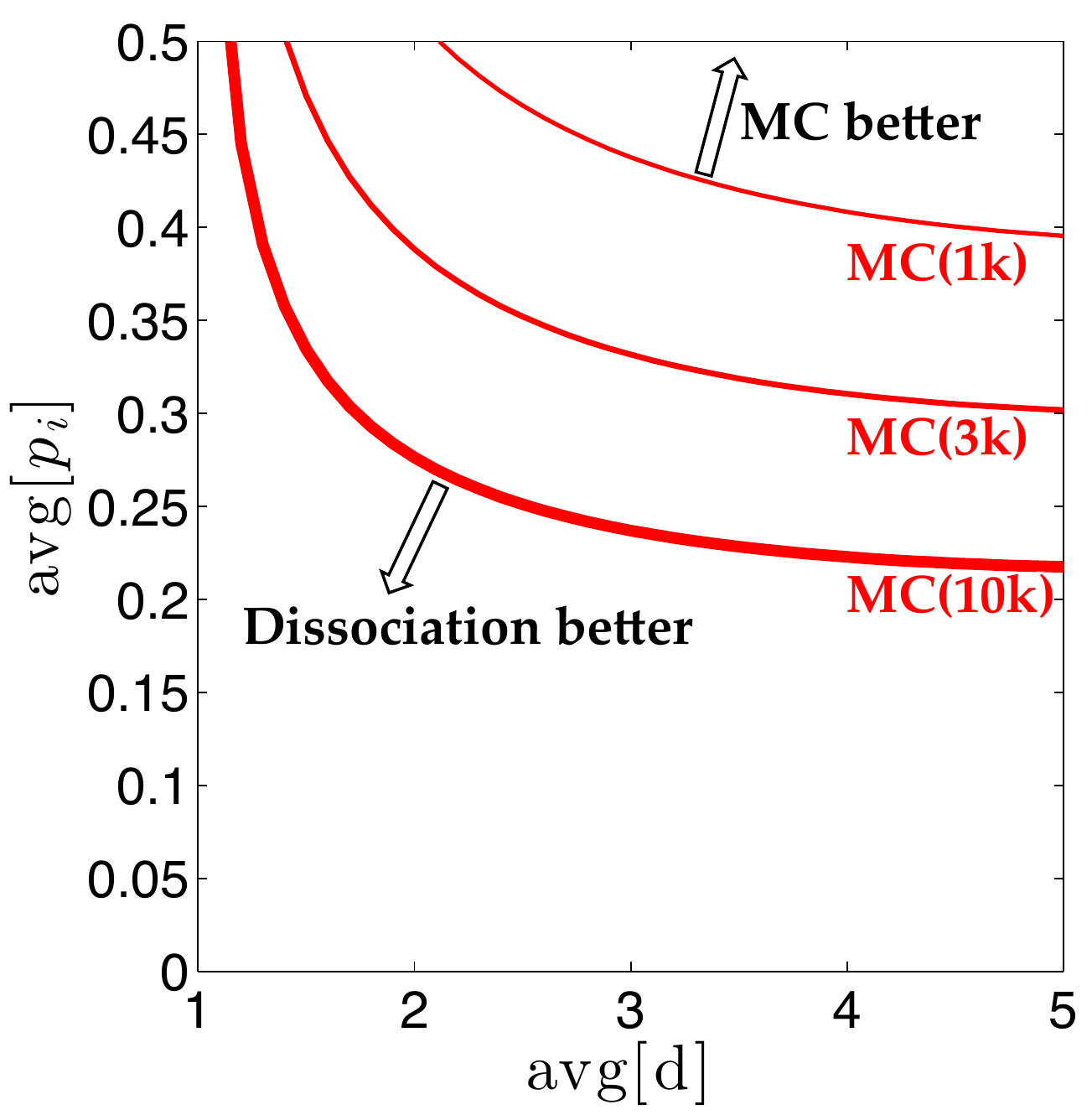}
		\label{Fig_VLDBJ_TPCH_AP_Diss_MC_Tradeoff2}}
	\hspace{-1mm} 
	\subfloat[\autoref{lesson6}]
		{\includegraphics[scale=0.325]{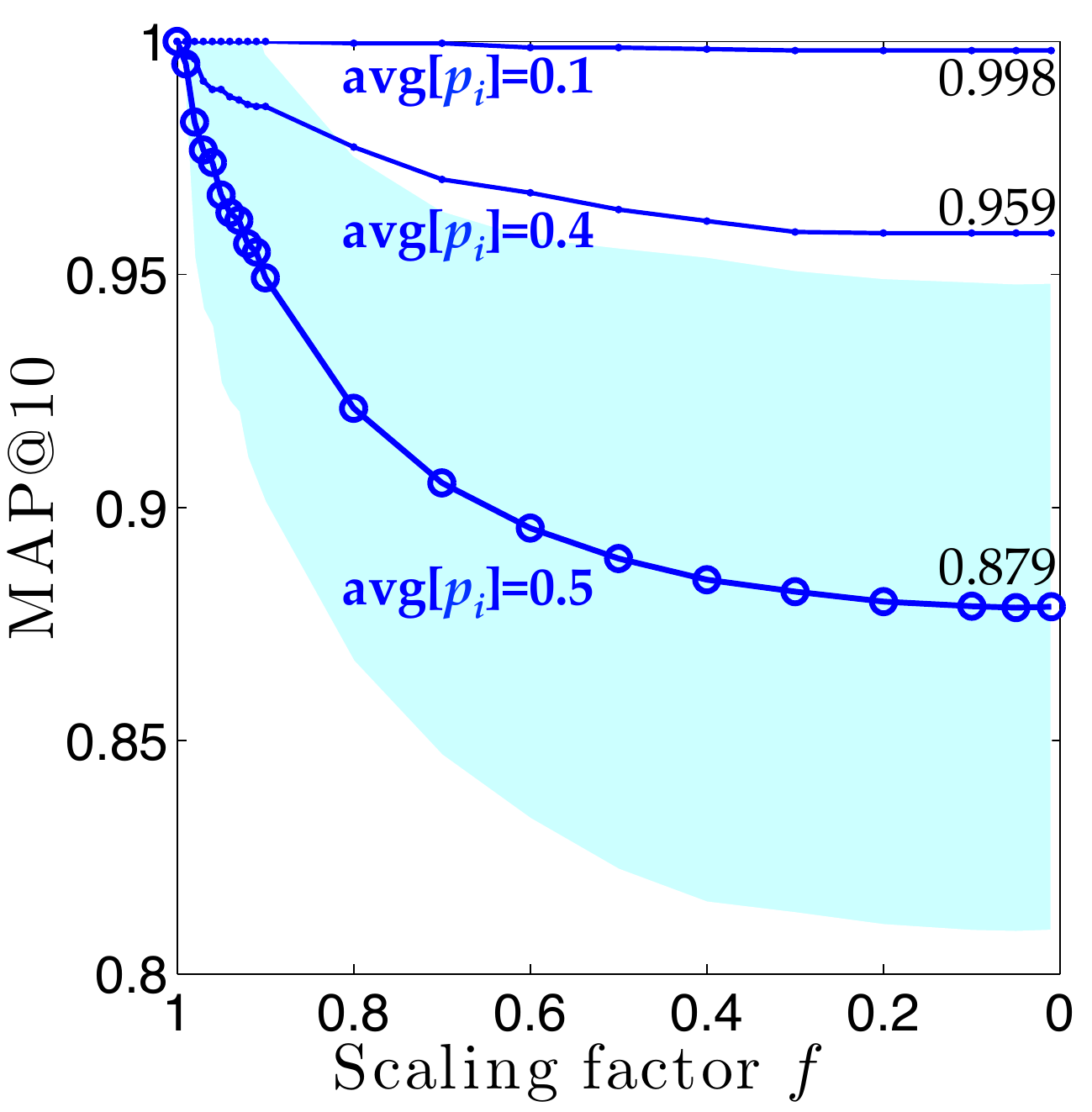}
		\label{Fig_VLDBJ_TPCH_AP_ScaledGT}}	
	\hspace{1mm} 
	\subfloat[\autoref{lesson6}]
		{\includegraphics[scale=0.38]{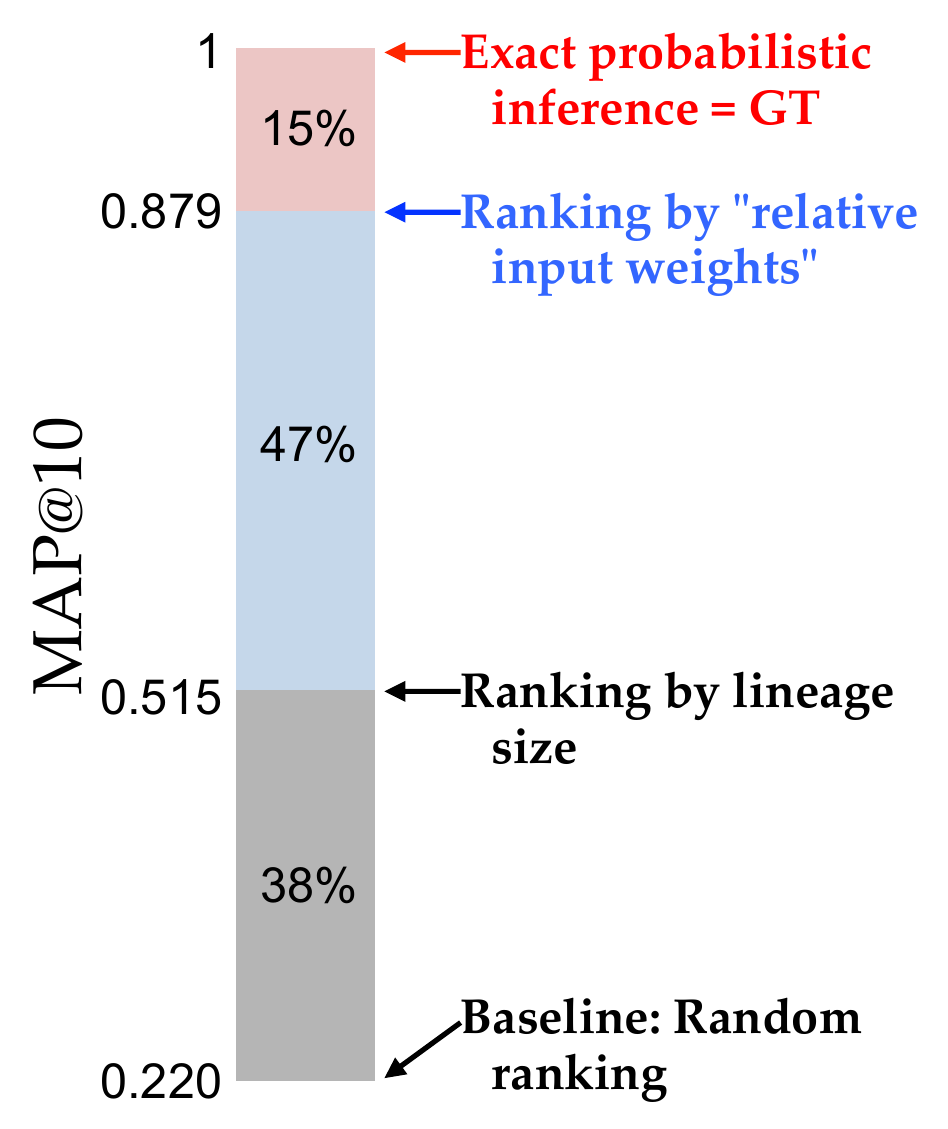}
		\label{Fig_VLDBJ_TPCH_AP_overviewBars}}	
	\hspace{3mm} 
	\subfloat[\autoref{lesson7}]
		{\includegraphics[scale=0.325]{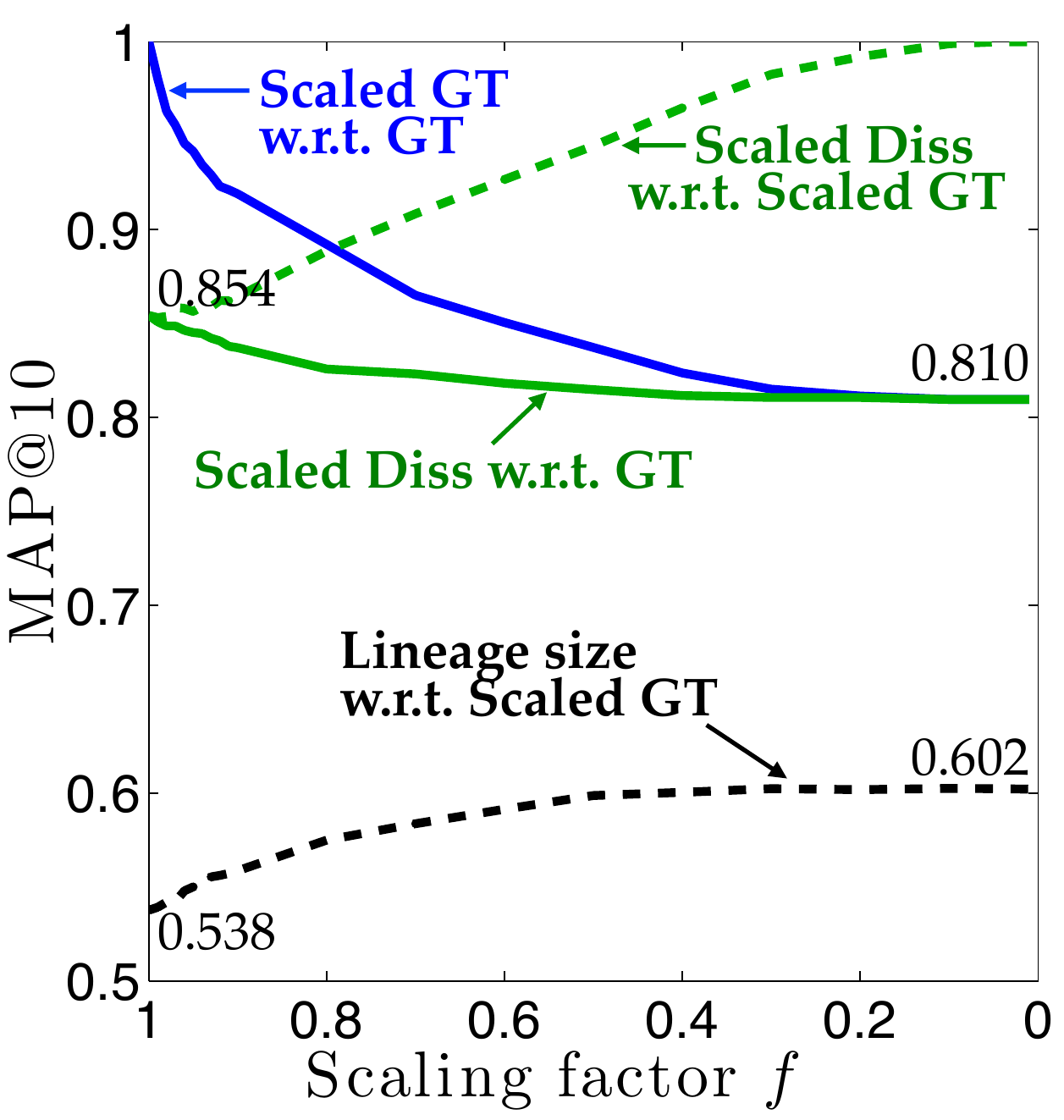}
		\label{Fig_VLDBJ_TPCH_AP_ScaledDissociation}}	
	\caption{
	Timing results: 
	(a)-(c) For increasing database sizes and constant cardinalities, our optimizations approach deterministic SQL performance.
	(d) Our optimizations can even evaluate very large number of minimal plans efficiently (here shown up to 429 for an 8-chain query).
	(e)-(h) For the TPC-H query, the best evaluation for dissociation is within a factor of $6$ of that for deterministic query evaluation.
	(i)-(p) Ranking experiments on TPC-H: Assumptions for each subfigure and conclusions that can be drawn are described in the main text in the respective result paragraph. }
    \label{Fig_VLDBJ_TPCH_AP_Diss_MC_Tradeoff}
\end{figure*}

\begin{question}\label{question8}
How does dissociation compare against other probabilistic methods and standard query evaluation?
\end{question}

\begin{result}\label{lesson8}
The best evaluation strategy for dissociation takes only a small overhead over standard SQL evaluation and is considerably faster than other probabilistic methods for large lineages.
\end{result}

\noindent
\Figuresref{Fig_SyntheticChainQuerySize}{Fig_VLDBJ_TPCH_timing4}
show that SampleSearch does not scale to larger lineages as 
the performance of exact probabilistic inference
depends on the
tree-width of the Boolean lineage formula, which generally increases with the size of the data. In contrast, dissociation is \emph{independent of the treewidth}.
For example, SampleSearch needed 780 sec for calculating the ground truth for a query with $\max[\textrm{lin}]= 5.9$k 
for which dissociation took 3.0 sec, and MC(1k) took 42 sec for a query with $\max[\textrm{lin}]= 4.2$k 
for which dissociation took 2.4 sec.
Dissociation takes only 10.5 sec for our largest query $\$2 = \sql{'\%'}$ and  $\$1 = 10k$ with  $\max[\textrm{lin}]= 35$k.
Retrieving the lineage for that query alone takes 5.8 sec, which implies that any probabilistic method 
that evaluates the probabilities outside of the database engine needs to issue this query to retrieve the DNF for each answer and would thus have to evaluate lineages of sizes around 35k in only 4.7 (= 10.5 - 5.8) sec to be faster than dissociation.\footnote{The time needed for the lineage query thus serves as minimum benchmark for \emph{any} probabilistic approximation.
The reported times for SampleSearch and MC are the sum of time for retrieving the lineage plus the actual calculations, without the time for reading and writing the input and output files for SampleSearch.}

\subsection{Ranking experiments}\label{sec:exRanking}

\noindent
For the following experiments, we are limited to those query parameters \$1 and \$2 for which we can get the ground truth (and results from MC) in acceptable time.
We systematically vary $p_{i\max}$ between $0.1$ and $1$ (and thus $\avg[p_{i}]$ between $0.05$ and $0.5$) and evaluate the rankings several times over randomly assigned input tuple probabilities.
We only keep data points (i.e.\ results of individual ranking experiments) for which the output probabilities are not too close to 1 to be meaningful ($\max[p_{a}]<0.999\,999$).

\begin{question}\label{question1}
How does ranking quality compare for our three ranking methods and which are the most important factors that determine the quality for each method?
\end{question}

\begin{result}\label{lesson1}
Dissociation performs better than MC which performs better than ranking by lineage size.
\end{result}

\noindent
\Autoref{Fig_VLDBJ_TPCH_AP_MC_aggregated_09} shows averaged results of our probabilistic methods for $\$2 = \sql{'\%red\%green\%'}$.\footnote{Results for MC with other parameters of \$2 are similar. However, the evaluation time for the experiments becomes quickly infeasible.} 
Shaded areas indicate standard deviations and the x-axis shows varying numbers of MC samples. 
We only used those data points for which $\avg[p_{a}]$ of the top 10 ranked tuples is between $0.1$ and $0.9$ according to ground truth 
($\approx 6$k data points for dissociation and lineage, $\approx 60$k data points for MC, as we repeated each MC simulation 10 times), as this is the best regime for MC, according to \autoref{lesson2}.
We also evaluated quality for dissociation and ranking by lineage for more queries by choosing parameter values for $\$2$ from a set of 28 strings, such as \sql{'\%r\%g\%r\%a\%n\%d\%'} 
and \sql{'\%re\%re\%'}.
The average MAP over all 28 choices for parameters \$2 is 0.997 for ranking by dissociation and 0.520 for ranking by lineage size ($\approx 100$k data points). Most of those queries have too large of a lineage to evaluate MC. Note that ranking by lineage always returns the same ranking for given parameters \$1 and \$2, but the GT ranking would change with different input probabilities.

\begin{result}\label{lesson2}
Ranking quality of MC increases with the number of samples and decreases when the average probability of the answer tuples $\avg[p_a]$ is close to $0$ or $1$.  
\end{result}	

\Autoref{Fig_VLDBJ_TPCH_ap_mc_random_redgreen_equalSizedBins_log}
shows the AP as a function of $\avg[p_{a}]$ of the top 10 ranked tuples according to ground truth by  logarithmic scaling of the x-axis (each point in the plot averages AP over $\approx 450$ experiments for dissociation and lineage and over $\approx 4.5$k experiments for MC). We see that MC performs increasingly poor for ranking answer tuples with probabilities close to $0$ or $1$ and even approach the quality of random ranking (MAP@10 = 0.22). 
This is so because, for these parameters, 
the probabilities of the top 10 answers are very close, and MC needs many iterations to distinguish them.
Therefore, MC performs increasingly poorly for increasing size of lineage but fixed average input probability $\avg [p_{i}] \approx 0.5$, as the average answer probabilities $\avg[p_{a}]$ will be close to $1$. 
In order not to ``bias against our competitor,'' we compared against MC
in its best regime with $0.1< \avg[p_a] <0.9$ in \autoref{Fig_VLDBJ_TPCH_AP_MC_aggregated_09}.

\begin{result}\label{lesson3}
Ranking by lineage size has good quality only when all input tuples have the same probability.
\end{result}

\Autoref{Fig_VLDBJ_TPCH_AP_LineageRanking} shows that ranking by lineage is good only when all tuples in the database have the \emph{same} probability
(labeled by $p_i = \textrm{const}$ as compared to $\avg[p_i] = \textrm{const}$). 
This is a consequence of the output probabilities depending mostly on the size of the lineages
if all probabilities are equal.
Dependence on other parameters, such as overall lineage size and magnitude of input probabilities (here shown for $p_i = 0.1$ and $p_i = 0.5$), seem to matter only slightly.

\begin{result}\label{lesson4}
The quality of dissociation decreases with the average number of dissociations per tuple $\avg[d]$ and with the average input probabilities $\avg[p_{i}]$.
Dissociation performs very well and notably better then MC(10k) if either $\avg[d]$ or $\avg[p_i]$ are small.
\end{result}

\noindent
Each answer tuple $a$ gets its score $p_a$ from one of two query plans $P_S$ and $P_P$ that dissociate tuples in tables $S$ and $P$, respectively. For example, if the lineage size for tuple $a$ is 100 and the lineage contains 20 unique suppliers from table $S$ and 50 unique parts from table $P$, then $P_S$ dissociates each tuple from $S$  into 5 tuples and $P_P$ each tuple from $P$ into $2$ tuples, on average. Most often, $P_P$ will then give the better bounds as it has fewer average dissociations. 
Let $\avg[d]$ be the mean number of dissociations for each tuple in the dissociated table of its respective optimal query plan, averaged across all top 10 ranked answer tuples. 
For all our queries (even those with $\$1 = 10k$ and $\$2 = \sql{'\%'}$), $\avg[d]$ stays below $1.1$ as, for each tuple, there is usually one plan that dissociates few variables. 
In order to \emph{understand the impact of higher numbers of dissociations} (increasing $\avg[d]$), we also measured AP for the ranking for \emph{each query plan individually}. 
Hence, for each choice of random parameters, we record two new data points -- one for ranking all answer tuples by using only $P_S$ and one for using only $P_P$ -- together with the values of $\avg[d]$ in the respective table that gets dissociated. This allows us to draw conclusions for a larger set of parameters.
\Autoref{Fig_VLDBJ_TPCH_AP_Diss_MC_Tradeoff1} plots MAP values 
as a function of $\avg[d]$ of the top 10 ranked tuples on the horizontal axis, and various values of $\avg[p_i]$ ($\avg[p_i] = 0.05, 0.10, \ldots, 0.5$).
Each plotted point averages over at least 10 data points (some have 10, other several 1000s). 
Dashed lines show a fitted parameterized curve to the data points on $\avg[p_i]$ and $\avg[d]$.
The figure also shows the standard deviations as shaded areas for $\avg[p_i] = 0.5$.
We see that the quality is very dependent on $\avg[p_i]$, as predicted by \autoref{prop:smallProbabilities}.

\Autoref{Fig_VLDBJ_TPCH_AP_Diss_MC_Tradeoff2} maps the trade-off between dissociation and MC for the two important parameters for the quality of dissociation ($\avg[d]$ and $\avg[p_i]$) and the number of samples for MC. For example, MC(1k) gives a better expected ranking than dissociation only for the small area above the thick red curve marked MC(1k). 
For MC, we used the test results from \autoref{Fig_VLDBJ_TPCH_AP_MC_aggregated_09}; i.e.\ assuming $0.1< \avg[p_a] <0.9$ for MC.
Also recall that for large lineages, having an input probability with $\avg[p_{i}] = 0.5$ will often lead to answer probabilities close to $1$ for which ranking is not possible anymore (recall \autoref{Fig_VLDBJ_TPCH_AP_LineageRanking}). Thus, for large lineages, we need small input probabilities to have meaningful interpretations. And for small input probabilities, dissociation considerably outperforms any other method.

\begin{question}\label{question6}
How much would the ranking change according to exact probabilistic inference if we scale down all input tuples?
\end{question}

\begin{result}\label{lesson6}
If the probabilities of all input tuples are already small, then scaling them further down does not affect the ranking much.
\end{result}

\noindent
Here, we repeatedly evaluated the exact ranking for 7 different parameterized queries over randomly generated databases with one query plan that has 
$\avg[d] \approx 3$, for two conditions:
first on a probabilistic database with $avg[p_{i}]$ input probabilities (we defined the resulting ranking as GT);
then again on a scaled version, where all input probabilities in the database are multiplied by the same scaling factor $f \in (0,1)$. We then compared the new ranking against GT. 
\Autoref{Fig_VLDBJ_TPCH_AP_ScaledGT} shows that if all input probabilities are already small (and dissociation already works well), then scaling has little 
effect on the ranking. 
However, for $avg[p_{i}]=0.5$ (and thus many tuples with $p_i$ close to $1$), we have a few tuples with $p_i$ close to $1$. These tuples are very influential for the final ranking, but their relative influence decreases if scaled down even slightly. 
Also note that even for $avg[p_{i}]=0.5$, scaling a database by a factor $f = 0.01$ instead of $f = 0.2$ does not make a big difference. However, the quality remains well above ranking by lineage size (!). 
This suggests that the difference between ranking by lineage size (MAP $\approx  0.529$) and the ranking on a scaled database for $f \rightarrow 0$ (MAP $\approx  0.879$) can be attributed to the relative weights of the input tuples (we thus refer to this as ``\emph{ranking by relative input weights}''). The remaining difference in quality then comes from the \emph{actual probabilities} assigned to each tuple.
Using MAP $\approx  0.220$ as baseline for random ranking, 
38\% of the ranking quality can be found by the lineage size alone vs.\ 85\% by the lineage size plus the relative weights of input tuples. The remaining 15\% come from the actual probabilities (\autoref{Fig_VLDBJ_TPCH_AP_overviewBars}).

\begin{question}\label{question7}
Does the expected ranking quality of dissociation decrease to random ranking for increasing fractions of dissociation (just like MC does for decreasing number of samples)?
\end{question}

\begin{result}\label{lesson7}
The expected performance of dissociation for increasing $\avg[d]$ for a particular query is lower bounded by the quality of ranking by relative input weights.
\end{result}

\noindent
Here, we 
use a similar setup as before and now compare various rankings against each other:
SampleSearch on the original database (``GT'');
SampleSearch on the scaled database (``Scaled GT'');
dissociation on the scaled database (``Scaled Diss''); and
ranking by lineage size (which is unaffected by scaling).
From \autoref{Fig_VLDBJ_TPCH_AP_ScaledDissociation}, we see
that the quality of Scaled Diss w.r.t.\ Scaled GT $\rightarrow 1$ for $f \rightarrow 0$ since dissociation works increasingly well for small $\avg[p_i]$ (recall \autoref{prop:smallProbabilities}). We also see that Scaled Diss w.r.t.\ GT decreases towards Scaled GT w.r.t.\ GT for $f \rightarrow 0$.
Since dissociation can always reproduce the ranking quality of ranking by relative input weights by first downscaling the database (though losing information about the actual probabilities) the expected quality of dissociation for smaller scales does not decrease to random ranking, but rather to ranking by relative weights. 
Note this result only holds for the expected MAP; any particular ranking can still be very much off.

%% file: 6relatedWork.tex
\section{Related Work}\label{sec:relatedWork}

\introparagraph{Probabilistic databases} Current approaches to query
evaluation on probabilistic databases can be classified into three
categories: 
($i$)~\emph{incomplete approaches} identify tractable
cases either at the
query-level~\cite{DBLP:journals/vldb/DalviS07,DBLP:journals/jacm/DalviS12,FinkO:PODS2014dichotomy}
or the
data-level~\cite{DBLP:conf/sum/OlteanuH08,DBLP:conf/icdt/RoyPT11,SenDeshpandeGetoor2010:ReadOnce};
($ii$)~\emph{exact approaches}
\cite{DBLP:conf/icde/AntovaJKO08,JhaOS2010:EDBT,DBLP:conf/sum/OlteanuH08,DBLP:conf/icde/OlteanuHK09,DBLP:conf/icde/SenD07}
work well on queries with simple lineage expressions, but perform
poorly on database instances with complex lineage expressions.
($iii$)~\emph{approximate approaches} either apply general purpose
sampling
methods~\cite{DBLP:conf/sigmod/JampaniXWPJH08,DBLP:journals/vldb/JoshiJ09,DBLP:conf/icde/KennedyK10,DBLP:conf/icde/ReDS07},
or approximate the number of models of the Boolean lineage
expression~\cite{DBLP:dblp_conf/icdt/FinkO11,OlteanuHK2010:ICDE,DBLP:journals/pvldb/ReS08}.
Our work can be seen as a generalization of  several
of these techniques: Our algorithm returns the exact score if the
query is
safe~\cite{DBLP:journals/vldb/DalviS07,DBLP:conf/icde/OlteanuHK09} or
data-safe~\cite{JhaOS2010:EDBT}.

\introparagraph{Lifted and approximate inference} Lifted inference was
introduced in the AI literature as an approach to probabilistic
inference that uses the first-order formula to exploit symmetries at
the grounded level~\cite{DBLP:conf/ijcai/Poole03}.  This research
evolved independently of that on probabilistic databases, and the two
have many analogies: A formula is called \emph{domain liftable} iff
its data complexity is in polynomial time~\cite{jaeger-broeck-2012},
which is the same as a \emph{safe query} in probabilistic databases,
and the FO-d-DNNF circuits described
in~\cite{DBLP:conf/ijcai/BroeckTMDR11} correspond to the safe plans
discussed in this paper.  See~\cite{BroeckSuciu:UAI2014tutorial} for a
recent discussion on the similarities and differences.

\introparagraph{Representing Correlations} The most popular approach
to represent correlations between  tuples in a probabilistic
database is by a Markov Logic network (MLN) which is a set of
\emph{soft constraints}~\cite{DBLP:series/synthesis/2009Domingos}.
Quite remarkably, all complex correlations introduced by an MLN can be
rewritten into a query over a tuple-independent probabilistic
database~\cite{DBLP:conf/uai/GogateD11a,BroeckMD:KR2014,DBLP:journals/pvldb/JhaS12}.
In combination with such rewritings, our techniques can be also
applied to MLNs if their rewritings results in conjunctive
queries without self-joins.

\introparagraph{Dissociation} Dissociation was first introduced in the
workshop paper~\cite{GatterbauerJS2010:MUD}, presented as a way to
generalize graph propagation algorithms to hypergraphs.  Theoretical
upper and lower bounds for dissociation of Boolean formulas, including
\autoref{th:bool:dissoc}, were proven
in~\cite{DBLP:journals/tods/GatterbauerS14}.
Dissociation is related to a technique called \emph{relaxation}
for probabilistic inference in graphical
models~\cite{DBLP:conf/uai/BroeckCD12}.

%% file: 7conclusion.tex
\section{Conclusions and Outlook}\label{sec:conclusion}
This paper proposes to approximate probabilistic query evaluation by evaluating a fixed number of query plans, each providing an upper bound on the true probability, then taking their minimum. 
We provide an algorithm that takes into account important schema information to enumerate only the minimal necessary plans among all possible plans, and prove it to be a strict generalization of all known results of PTIME self-join-free conjunctive queries. 
We describe relational query optimization techniques that 
allow us to evaluate all minimal queries in a single query and very fast:
Our experiments show that these optimizations bring approximate probabilistic query evaluation close to standard query evaluation while providing high ranking quality.
In future work, we plan to generalize this approach to full first-order queries.
We will also make slides illustrating our algorithms available at \url{http://LaPushDB.com}.